\documentclass[journal, a4paper]{IEEEtran}
\usepackage[left=.75in,right=.75in,top=1in,bottom=.75in]{geometry}
\usepackage{vkmacros}
\graphicspath{{figures/}}
\newif\ifshowtodo
\showtodotrue

\newif\ifshowdetail
\showdetailfalse

\newif\ifvlong
\vlongfalse

\newcommand{\VersionLength}{long}
\providecommand{\verlong}{\ifthenelse{\equal{\VersionLength}{long}}}
\newcommand{\VersionCols}{single}
\providecommand{\dcol}{\ifthenelse{\equal{\VersionCols}{double}}}

\long\def\detail#1{\ifshowdetail{{\color{green!50!black}#1}}\fi}

\makeatletter
\newcommand{\subalign}[1]{%
  \vcenter{%
    \Let@ \restore@math@cr \default@tag
    \baselineskip\fontdimen10 \scriptfont\tw@
    \advance\baselineskip\fontdimen12 \scriptfont\tw@
    \lineskip\thr@@\fontdimen8 \scriptfont\thr@@
    \lineskiplimit\lineskip
    \ialign{\hfil$\m@th\scriptstyle##$&$\m@th\scriptstyle{}##$\hfil\crcr
      #1\crcr
    }%
  }%
}
\makeatother

\def\D{\mathcal{D}} 

\def\noarrow{\ar@{-}[r]}

\title{The CEO problem with inter-block memory}
\author{
Victoria Kostina,~\IEEEmembership{Member,~IEEE}, Babak Hassibi,~\IEEEmembership{Member,~IEEE}
\thanks{
The authors are with California Institute of Technology (e-mail: \href{mailto:vkostina@caltech.edu}{vkostina@caltech.edu}, \href{mailto:hassibi@caltech.edu}{hassibi@caltech.edu}). 
This work was supported in part by the National Science Foundation (NSF)
under grants CCF-1751356 and CCF-1817241. The work of Babak Hassibi was supported in part by the NSF under grants CNS-0932428, CCF-1018927, CCF-1423663 and CCF-1409204, by a grant from Qualcomm Inc., by NASA's Jet Propulsion Laboratory through the President and Director's Fund, and by King Abdullah University of Science and Technology. A part of this work was presented at ISIT 2020~\cite{kostina2020fundamentalISIT}.
}}
\IEEEoverridecommandlockouts

\date{}							
\begin{document}
\maketitle

\begin{abstract}
An $n$-dimensional source with memory is observed by $K$ isolated encoders via parallel channels, who compress their observations to transmit to the decoder via noiseless rate-constrained links while leveraging their memory of the past. At each time instant, the decoder receives $K$ new codewords from the observers, combines them with the past received codewords, and produces a minimum-distortion estimate of the latest block of $n$ source symbols. This scenario extends the classical one-shot CEO problem to multiple rounds of communication with communicators maintaining the memory of the past. 

We extend the Berger-Tung inner and outer bounds to the scenario with inter-block memory, showing that the minimum asymptotically (as $n \to \infty$) achievable sum rate required to achieve a target distortion is bounded by minimal directed mutual information problems. For the Gauss-Markov source observed via $K$ parallel AWGN channels, we show that the inner bound is tight and solve the corresponding minimal directed mutual information problem, thereby establishing the minimum asymptotically achievable sum rate. Finally, we explicitly bound the rate loss due to a lack of communication among the observers; that bound is attained with equality in the case of identical observation channels. 

The general coding theorem is proved via a new nonasymptotic bound that uses stochastic likelihood coders and whose asymptotic analysis yields an extension of the Berger-Tung inner bound to the causal setting. The analysis of the Gaussian case is facilitated by reversing the channels of the observers. 
\end{abstract}


\section{Introduction}
We set up the CEO (chief executive or estimation officer) problem with inter-block memory as follows. An information source $\{X_i\}$ emits a block of length $n$, $X_i \in \mathcal A^n$, at time $i$; it is observed by $K$ encoders through $K$ noisy channels; at time $i$, $k$th encoder sees $Y_i^k$ generated according to $P_{Y_i^k | X_1, \ldots, X_{i}, Y_{1}^k, \ldots, Y_{i-1}^k}$. See \figref{fig:system}.
The encoders (observers) communicate to the decoder (CEO) via their separate noiseless rate-constrained links. At each time $i$, $k$th observer forms a codeword based on the observations it has seen so far, i.e., $Y_1^k, \ldots, Y_i^k$. 
The decoder at time $i$ forms the estimate, $\hat X_i \in \hat {\mathcal A}^n$, based on the codewords it received thus far. The goal is to minimize the average distortion 
\begin{equation}
 \frac 1 t \sum_{i = 1}^{t}  \E{ \sd (X_i, \hat X_i)} \label{eq:dintro},
\end{equation}
where $t$ is the \emph{time horizon} over which the source is being tracked, and $\sd \colon \mathcal A^n \times \hat {\mathcal A}^n \mapsto \mathbb R_+$ is the distortion measure. Encoding and decoding operations leverage the memory of the past but cannot look in the future. In this causal setting no delay is allowed neither at the encoders in producing codewords to encode $X_i$ nor at the decoder in producing~$\hat X_i$.  

\vspace{5pt}
\begin{figure}[htp]
\begin{center}
    \includegraphics[width=1\linewidth]{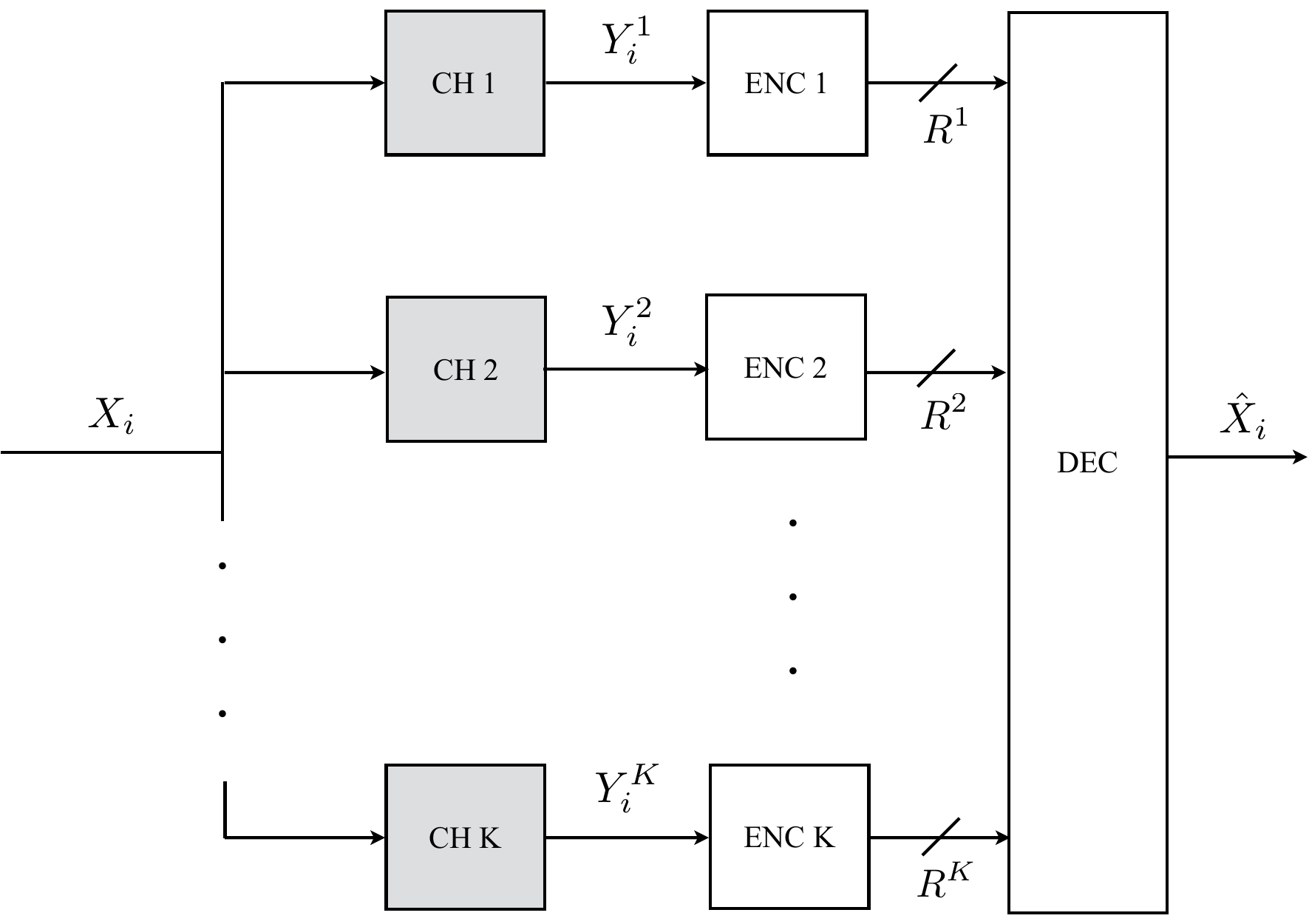}
\end{center}
 \caption[]{The CEO problem with inter-block memory: the encoders and the decoder keep the memory of their past observations.} 
 \label{fig:system}
\end{figure}

In the classical setting with $t = 1$, the CEO problem was first introduced by Berger et al.~\cite{berger1996ceo} for a finite alphabet source. In the classical Gaussian CEO problem, an i.i.d. Gaussian source is observed via AWGN channels and reproduced under mean squared error (MSE) distortion. The Gaussian CEO problem was studied by Viswanathan and Berger~\cite{viswanathan1997quadratic}, who proved an achievability bound on the rate-distortion dimension for the case of $K$ identical Gaussian channels, by Oohama~\cite{oohama1998ceo}, who derived the sum-rate rate-distortion region for that special case, by Prabharan et al. \cite{prabhakaran2004ceo} and Oohama~\cite{oohama2005rate}, who determined the full Gaussian CEO rate region, by Chen et al. \cite{chen2004upper}, who proved that the minimum sum rate is achieved via waterfilling, by Behroozi and Soleymani \cite{behroozi2009optimal} and by Chen and Berger \cite{chen2008successive}, who showed rate-optimal successive coding schemes. 
 Wagner et al. \cite{wagner2008rate} found the rate region of the distributed Gaussian lossy compression problem by coupling it to the Gaussian CEO problem. Wagner and Anantharam \cite{wagner2008improved} showed an outer bound to the rate region of the multiterminal source coding problem that is tighter than the Berger-Tung outer bound \cite{berger1978multi,tung1978multiterminal}. Wang et al. \cite{wang2010sum} showed a simple converse on the sum rate of the vector Gaussian CEO problem. 
Concurrently, Ekrem and Ulukus \cite{ekrem2014outer} and Wang and Chen \cite{wang2014vector}  showed an outer bound to the rate region of the vector Gaussian CEO problem that is tight in some cases and not tight in others and that particularizes the outer bound in \cite{wagner2008improved} to the Gaussian case. Courtade and Weissman \cite{courtade2014multiterminal} determined the
distortion region of the distributed source coding and the CEO problem  under logarithmic loss. 

None of the above results directly apply to the tracking problem in \figref{fig:system} because of the past memory in encoding the $n$-blocks of observations and in producing $\hat X_i$ in \eqref{eq:dintro}, which imposes blockwise causality constraints onto the coding process. 
The most basic scenario of source coding with causality constraints is that of a single observer directly seeing the information source \cite{gorbunov1973nonanticipatory}.  The causal rate-distortion function for the Gauss-Markov source was computed by Gorbunov and Pinsker \cite{gorbunov1974prognostic}. The link between the minimum attainable linear quadratic Gaussian (LQG) control cost and the causal rate-distortion function is elucidated in \cite{tatikonda2004stochastic,silva2016characterization,kostina2016ratecost}.  A semidefinite program to compute the causal rate-distortion function for vector Gauss-Markov sources is provided in~\cite{tanaka2017semidefinite}. The remote Gaussian causal rate-distortion function, which corresponds to setting $K = 1$ in \figref{fig:system}, is computed in \cite{kostina2016ratecost}. The causal rate-distortion function of the Gauss-Markov source with a Gaussian side observation available at the decoder (the causal counterpart of the Wyner-Ziv setting) is computed in \cite{kostina2018SIallerton} for the scalar source and in \cite{sabag2020minimal} for the vector source. That causal Wyner-Ziv setting can be viewed a special case of our causal CEO problem \eqref{eq:xi}, \eqref{eq:yik} with two observers, with one of the observers enjoying an infinite rate. Stability of linear Gaussian systems with multiple isolated observers is investigated in \cite{johnston2014stochastic}.

The first contribution of this paper is an extension of the Berger-Tung inner and outer bounds \cite{berger1978multi,tung1978multiterminal} to the distributed tracking setting of \figref{fig:system} that sandwich the minimum asymptotically achievable (as $n \to \infty$)  sum rate $R^1 + \ldots + R^K$ required to achieve a given average distortion \eqref{eq:dintro}. Provided that the components of each $X_i \in \mathcal A^n$ are i.i.d. ($X_i$ can still depend on $X_1, \ldots, X_{i-1}$), the channels act on each of those components independently, and the distortion measure is separable,
that minimum sum rate is bounded in terms of the directed mutual information from the encoders to the decoder. The converse (outer bound) follows via standard data processing and single-letterization arguments. To prove the achievability, we show a nonasymptotic bound for blockwise-causal distributed lossy source coding that can be viewed as an extension of the nonasymptotic Berger-Tung inner bound by Yassaee et al. \cite{yassaee2013technique,yassaee2013techniqueArxiv}, applicable to the setting with $K = 2$ sources and $t = 1$ rounds of communication, to the setting with an arbitrary number of sources and communication rounds. We view the horizon-$t$ causal coding problem as a multiterminal coding problem in which at each step coded side information from past steps is available, and we use a stochastic likelihood coder (SLC) by Yassaee et al. \cite{yassaee2013technique,yassaee2013techniqueArxiv} to perform encoding operations.  The SLC-based encoder mimics the operation of the joint typicality encoder while admitting sharp nonasymptotic bounds on its performance. While the SLC-based decoder of \cite{yassaee2013technique,yassaee2013techniqueArxiv} is ill-suited to the case $K > 2$, we propose a novel decoder that falls into the class of generalized likelihood decoders \cite{merhav2017gld} and uses $K$ different threshold tests depending on the point of the rate-distortion region the code is operating at. An asymptotic analysis of our nonasymptotic bound yields an extension of the Berger-Tung inner bound  \cite{berger1978multi,tung1978multiterminal} to the setting with inter-block memory.

The second contribution of the paper is an explicit evaluation of the minimum sum rate for the causal Gaussian CEO problem.  In that scenario, the source is an $n$-dimensional Gauss-Markov source, 
\begin{align}
X_{i+1} &= a X_i + V_i,  \label{eq:xi}
\end{align}
and the $k$-th observer sees
\begin{align}
Y_{i}^k &=  X_i + W_{i}^k, \quad k = 1, \ldots, K, \label{eq:yik}
\end{align}
where $X_1$ and $\{V_i, W_{i}^1, W_{i}^2, \ldots, W_{i}^K\}_{i = 1}^T$ are independent Gaussian vectors of length $n$ with i.i.d. components; each component of $V_i$ is distributed as $ \mathcal N(0, \sigma_{\mathsf V}^2 )$, and each component of $W_i^k$ as $\mathcal N(0, \sigma_{\mathsf W_k}^2)$. Note that different observation channels can have different noise powers. The distortion measure is the normalized squared error
\begin{equation}
 \mathsf d \left(X_i, \hat X_i\right) = \frac 1 n \|X_i - \hat X_i\|^2. \label{eq:MSEdef}
\end{equation}
We characterize the minimum sum rate as a convex optimization problem over $K$ parameters; an explicit formula is given in the case of identical observation channels.
Similar to the corresponding result for $t = 1$ \cite{prabhakaran2004ceo,oohama2005rate},\cite[Th. 12.3]{el2011network}, our extension of the Berger-Tung inner bound is tight in this case.  To compute the bound, we split up the directed minimal mutual information problem into a sum of easier-to-solve optimization problems. To tie the parameters of those optimization problems back to those of the original optimization problem, we extend the technique developed by Wang et al. \cite{wang2010sum} for the time horizon $t = 1$, to $t > 1$. A device that helps us track the behavior of optimal estimation errors over multiple time instances is the reversal of the channels from $\{X_i\}$ to $\{Y_i^k\}$:
\begin{align}
X_i = \bar X_i^k + W_i^{k \prime},
\end{align}
where
 \begin{align}
 \bar X_i^k \triangleq \E{X_i | Y_{1}^k, \ldots, Y_{i}^k}, \label{eq:Xbark}
\end{align}
and $W_i^{k\, \prime} \perp \bar X_i^k$ are Gaussian independent random vectors representing the errors in estimating $X_i$ from $\{Y_{j}^k\}_{j = 1}^i$. While for $t = 1$, it does not matter whether the encoders compress $Y_1^k$ or $\bar X_1$ since the latter is just a scaled version of the former, for $t > 1$, compressing $Y_i^k$ instead of $\bar X_i^k$ is only suboptimal.

The third contribution of the paper is a bound on the rate loss due to a lack of communication among the different encoders in the causal Gaussian CEO problem: as long as the target distortion is not too small, the rate loss is bounded above by $K-1$ times the difference between the remote and the direct rate-distortion functions. The bound is attained with equality if the observation channels are identical, indicating that among all possible observer channels with the same minimum MSE in the estimation of $\{X_i\}$ from $\{Y_j^{k}\}_{j \leq i, k = 1, \ldots, K}$, the identical channels case is the hardest to compress. This result contributes to the discussions of the rate loss in the classical CEO \cite[Cor.~1]{kostina2019ratelossITW} and multiple descriptions \cite[Lemma 3]{ostergaard2011incremental} problems.

The rest of the paper is organized as follows. In \secref{sec:dir}, we consider the general (non-Gaussian) causal CEO problem and prove direct and converse bounds to the minimum sum rate in terms of minimal directed mutual information problems (\thmref{thm:cg}). In \secref{sec:rd}, we characterize the causal Gaussian CEO rate-distortion function (\thmref{thm:causalceo}). In \secref{sec:loss}, we bound the rate loss due to isolated observers (\thmref{thm:loss}). 

\emph{Notation:} 
Logarithms are natural base. For a natural number $M$, $[M] \triangleq \{1, \ldots, M\}$. Notation $X \leftarrow Y$ reads ``replace $X$ by $Y$"; notation $X \perp Y$ reads ``$X$ is independent of $Y$''; notation $\triangleq$ reads ``by definition''. The temporal index is indicated in the subscript and the spatial index in the superscript: $Y_{[t]}^k$ is the temporal vector $(Y_1^k, \ldots, Y_t^k)$; $Y_i^{[K]}$  is the spatial vector $(Y_i^{1}, \ldots, Y_{i}^K)^{\mathsf T}$; $Y_{[t]}^{[K]} \triangleq (Y_{[t]}^1, \ldots, Y_{[t]}^K)$. Delay operator $\mathcal D$ acts as $\mathcal D X_{[t]} \triangleq (0, X_1, \ldots, X_{t-1})$. 
For a random vector $X$ with i.i.d. components, $\mathsf X$ denotes a random variable distributed the same as each component of $X$. 
We adopt the following shorthand notation for causally conditional \cite{kramer1998PhD} probability kernels: 
 \begin{equation}
 P_{Y_{[t]} || X_{[t]}} \triangleq \prod_{i=1}^{t} P_{Y_i | Y_{[i-1]}, X_{[i]}} \label{eq:causalcond}.
\end{equation}
Given a distribution $P_{X_{[t]}}$ and a causal kernel $P_{Y_{[t]} \| X_{[t]}}$, the directed mutual information is defined as \cite{massey1990causality}
\begin{equation}
I\left(X_{[t]} \to Y_{[t]}\right) \triangleq \sum_{i = 1}^{t}  I\left(X_{[i]}; Y_i | Y_{[i-1]}\right). \label{eq:Idir}
\end{equation}

\section{Sum rate via directed information}
\label{sec:dir}

\subsection{Overview}
In this section, we present and prove our extension of the Berger-Tung bounds to the setting inter-block memory that sandwich the minimum achievable sum rate in terms of minimal directed mutual information problems. The bounds apply to an abstract source with abstract observations. The operational scenario and achievable rates are formally defined in \secref{sec:oper}. The directed mutual information bounds are presented in \secref{sec:coding}. The converse is proven in \secref{sec:codingc}. The nonasymptotic achievability bound and its asymptotic analysis are presented in \secref{sec:a}. A set of remarks in  \secref{sec:remarks} completes \secref{sec:dir}. 
\subsection{Operational problem setting}
\label{sec:oper}
A CEO code with inter-block memory, or a causal CEO code, is formally defined as follows. 
\begin{defn}[A CEO code with inter-block memory]
Consider a discrete-time random process $\{X_i\}_{i = 1}^t$ on $\mathcal X$, observed by $K$ causal observers via the channels 
\begin{align}
P_{Y_{[t]}^{k} \| X_{[t]}} \colon \mathcal X^{\otimes t} \mapsto \mathcal Y^{\otimes t}, \quad  k \in [K]. \label{eq:channels}
\end{align}
Let $\mathsf d\colon \mathcal X \times \hat {\mathcal X} \mapsto \mathbb R_+$ be the distortion measure.

A CEO code with inter-block memory consists of: 
\begin{enumerate}[a)]
\item $K$ \emph{encoding policies}
\begin{align}
 P_{B_{[t]}^k \| Y_{[t]}^k}\colon  \mathcal Y^{\otimes t} \mapsto \prod_{i = 1}^t [M_i^k], \quad k \in [K],
 \label{eq:encpolicy}
\end{align}
\item a \emph{decoding policy} 
\begin{align}
 P_{ \hat X_{[t]}^{[K]} \| B_{[t]}^{[K]}} \colon \prod_{i = 1}^t [M_i^k] \mapsto \hat {\mathcal X}^{\otimes t}.
\end{align}
\end{enumerate}

 If the encoding and decoding policies satisfy
\begin{align}
& \frac 1 {t}  \sum_{i = 1}^t \E{ \mathsf d \left(X_i, \hat X_i\right)} \leq d, \label{eq:d}
\end{align}
we say that they form an $(M_{[t]}^{[K]}, d)$ \emph{average distortion} code.

If the encoding and decoding policies satisfy
\begin{align}
 \Prob{ \bigcup_{i = 1}^t \left\{ \mathsf d\left( X_i, \hat X_i \right)  > d_i \right\} } \leq \epsilon, \label{eq:excessdistceo}
\end{align} 
we say that they form an $(M_{[t]}^{[K]}, d_{[t]}, \epsilon)$ \emph{excess distortion} code.

The probability measure in \eqref{eq:d} and \eqref{eq:excessdistceo} is generated by the joint distribution $P_{X_{[t]}} P_{Y_{[t]}^{[K]} \| X_{[t]}} P_{ \hat X_{[t]}^{[K]} \| B_{[t]}^{[K]}} \prod_{k = 1}^K P_{B_{[t]}^k \| Y_{[t]}^k}$.

\label{defn:ceocode}
\end{defn}

A distortion measure $\mathsf d_n \colon \mathcal A^n \times \hat {\mathcal A}^n \mapsto \mathbb R_+$ is called \emph{separable} if 
\begin{align}
\mathsf d_n ( x, \hat x ) = \frac 1 n \sum_{i = 1}^n \sd (x (i), \hat x (i)), 
\end{align}
where $\mathsf d \colon \mathcal A \times \hat {\mathcal A} \mapsto \mathbb R_+$, and $x (i),~ \hat x (i)$ denote the $i$-th components of vectors $x \in \mathcal A^n$ and $\hat {x} \in\hat{ \mathcal A}^n $, respectively. 

\begin{defn}[Operational rate-distortion function]
Consider a discrete-time random process $\{X_i\}_{i = 1}^t$ on $\mathcal X = \mathcal A^n $ equipped with a separable distortion measure, observed by $K$ causal observers via the channels \eqref{eq:channels}.

 The rate-distortion tuple $\left(R^{[K]}, d\right)$ is \emph{asymptotically achievable} at time horizon $t$ if for $\forall \gamma > 0$, $\exists n_0 \in \mathbb N$ such that $\forall n \geq n_0$, an $\left(M_{[t]}^{[K]}, d + \gamma \right)$ average distortion CEO code with inter-block memory exists, where
\begin{align}
 \frac 1 {nt} \sum_{i = 1}^t \log M_i^k &\leq R^k, \quad k \in [K].\label{eq:sumlogMik}
\end{align}

The rate-distortion pair $\left(R, d\right)$ is asymptotically achievable if a rate-distortion tuple $\left(R^{[K]}, d\right)$ with
\begin{align}
\sum_{k = 1}^K R^k \leq R \label{eq:sumrateconstraint}
\end{align}
is asymptotically achievable. 

The\emph{ causal CEO rate-distortion function} at time horizon $t$ is defined as follows: 
\begin{align}
\!\!\!\!\!\! R_{t\,\mathrm{CEO}}(d) \triangleq \inf \Big\{ & R \colon  \left(R, d\right)  \text{ is achievable } 
\\
& \text{at time horizon $t$ in the CEO problem}. 
\Big\}  \notag
\end{align}
\label{def:Rceot}
\end{defn}

\subsection{Berger-Tung bounds with inter-block memory}
\label{sec:coding}

Consider a discrete-time random process $\{X_i\}_{i = 1}^t$ on $\mathcal X = \mathcal A^n $ equipped with separable distortion measure $\mathsf d$, observed by $K$ causal observers via the channels \eqref{eq:channels} with $\mathcal Y = \mathcal B^n$ and
\begin{align}
P_{X_i | X_{[i - 1]}} &= P_{\mat X_i | \mat X_{[i-1]}}^{\otimes n} \label{eq:singlelettersource} \\
P_{Y_{i}^{k} | X_{[i]}, Y_{[i-1]}^k} &= P_{\mat Y_i^k | \mat X_{[i]}, \mat Y_{[i-1]}^k }^{\otimes n}. \label{eq:singleletterchannel}
\end{align}

Denote the minimal directed mutual information problems
\begin{align}
\overline{\mathsf R}_{t\,\mathrm{CEO}}(d) &\triangleq \inf_
 {\subalign{
 P_{\mat U_{[t]}^{[K]} \| {\mat Y}_{[t]}^{[K]}} &\colon \eqref{eq:sepenc1} \\
 P_{\hat {\mat X}_{[t]}\| \mat U_{[t]}^{[K]}} &\colon \eqref{eq:dt1}
 } }
  \frac 1 t  I \left(  {\mat Y}^{[K]}_{[t]} \to  \mat U^{[K]}_{[t]} \right) \label{eq:bti}
  \\
  \underline {\mathsf R}_{t\,\mathrm{CEO}}(d) &\triangleq \inf_
 {\subalign{
 P_{\mat U_{[t]}^{[K]} \| {\mat Y}_{[t]}^{[K]}} &\colon \eqref{eq:bto1} \\
 P_{\hat {\mat X}_{[t]}\| \mat U_{[t]}^{[K]}} &\colon \eqref{eq:dt1}
 } }
  \frac 1 t I \left(  {\mat Y}^{[K]}_{[t]} \to  \mat U^{[K]}_{[t]} \right) \label{eq:bto}
\end{align}
where the constraints are as follows: 
\begin{align}
P_{\mat U_{[t]}^{[K]} \| {\mat Y}_{[t]}^{[K]}} &= \prod_{k = 1}^K P_{\mat U_{[t]}^k \|  {\mat  Y}_{[t]}^k} \label{eq:sepenc1}\\
P_{\mat U_{[t]}^{k} \| {\mat Y}_{[t]}^{[K]}} &= P_{\mat U_{[t]}^k \|  {\mat  Y}_{[t]}^k} \quad \forall k \in [K] \label{eq:bto1}\\
 \frac 1 {t}  \sum_{i = 1}^t \E{ \mathsf d \left(\mat X_i, \hat {\mat X}_i\right)} &\leq d. \label{eq:dt1}
\end{align}
Fixing a $k \in [K]$ and marginalizing $\{U_{[t]}^{k^\prime},~k^\prime \neq k\}$ out of both sides of \eqref{eq:sepenc1}, one can see that any joint distribution that satisfies the separate encoding constraint \eqref{eq:sepenc1}  also satisfies \eqref{eq:bto1}. Thus, the optimization problems \eqref{eq:bti} and \eqref{eq:bto} differ in that the constraint \eqref{eq:sepenc1} is more stringent than \eqref{eq:bto1}. They represent extensions of the Berger-Tung inner (\eqref{eq:bti}) and outer (\eqref{eq:bto}) bounds \cite[Th. 12.1, 12.2]{el2011network} to the causal setting. 

One can convexify $\overline{\mathsf R}_{t\,\mathrm{CEO}}(d)$ by adding to the optimization parameters a scalar $\alpha \in (0, 1]$ and a distribution $P_{\tilde{\mat U}_{[t]}^{[K]} \| \tilde{\mat Y}_{[t]}^{[K]}}$ satisfying the separate encoding constraint analogous to \eqref{eq:sepenc1}, and replacing the directed information in \eqref{eq:bti} by  $\alpha I \left(  {\mat Y}^{[K]}_{[t]} \to  \mat U^{[K]}_{[t]} \right) + (1 - \alpha) I \left(  \tilde {\mat Y}^{[K]}_{[t]} \to  \tilde{\mat U}^{[K]}_{[t]} \right)$. This is equivalent to introducing into \eqref{eq:bti} a binary time sharing random variable. Given the achievability of \eqref{eq:bti}, the achievability of the convexification follows by the standard time sharing argument \cite[Ch. 4.4]{el2011network}. 

Since a mixture of distributions $P_{\mat U_{[t]}^{[K]} \| {\mat Y}_{[t]}^{[K]}}$ satisfying \eqref{eq:bto1} also satisfies \eqref{eq:bto1}, the convexity of $\underline{\mathsf R}_{t\,\mathrm{CEO}}(d)$ follows from the convexity of directed mutual information in $P_{\mat U_{[t]}^{[K]} \| {\mat Y}_{[t]}^{[K]}}$, with no need for an explicit auxiliary time sharing random variable.

\begin{thm}[Berger-Tung bounds with inter-block memory]
Consider a discrete-time random process $\{X_i\}_{i = 1}^t$ on $\mathcal X = \mathcal A^n $ equipped with a separable distortion measure $\mathsf d$, observed by $K$ causal observers via the channels \eqref{eq:channels} with $\mathcal Y = \mathcal B^n$ and \eqref{eq:singlelettersource}, \eqref{eq:singleletterchannel} satisfied. 
Suppose further that for some $p > 1$, there exists a vector $\hat {\mathsf x}_{[t]}$ such that 
\begin{align}
 \left(\E{\left( \frac 1 t \sum_{i = 1}^t \sd(\mathsf X_i, \hat{\mat x}_i)\right)^p}\right)^{\frac 1 p} \leq d_p < \infty. \label{eq:dmom}
\end{align}
The causal rate-distortion function is bounded as
\begin{align}
\underline{\mathsf R}_{t\,\mathrm{CEO}}(d) \leq R_{t\,\mathrm{CEO}}(d)  \leq \overline{\mathsf R}_{t\,\mathrm{CEO}}(d)
   \label{eq:cdp}.
\end{align}
\label{thm:cg}
\end{thm}
Condition \eqref{eq:dmom} is a technical condition needed to apply a standard argument using H\"older's inequality to pass from an excess to average distortion in the proof of the achievability bound (Appendix~\ref{sec:btas}).
 
 To prove the upper bound on the sum rate in \eqref{eq:cdp}, we actually show a more accurate characterization of the entire rate tuple $R^{[K]}$ (\thmref{thm:btas}, below).
 
We will see in \secref{sec:rd} below that the inner (upper) bound in \eqref{eq:cdp} is tight in the quadratic Gaussian setting. This is in line with the corresponding result in the setting of block coding without inter-block memory \cite[Th. 12.3]{el2011network}. 

While in general the $t$-step optimization problems \eqref{eq:bti} and \eqref{eq:bto} are challenging to compute, we illustrate in this paper that the normalized limit as $t \to \infty$ is possible to compute in the Gaussian setting. Similar limit results in other communication scenarios were shown in \cite{gorbunov1974prognostic,tanaka2015stationary,kostina2016ratecost,kostina2018SIallerton,guo2019wienersampling,sabag2020minimal,sabag2021feedbackISIT,sabag2021feedback}.

\subsection{ \thmref{thm:cg}: proof of converse}
\label{sec:codingc}
The proof of the converse uses standard techniques. We will use the following definition and lemma. 

Causally conditioned directed information is defined as
\begin{equation}
I(X_{[t]} \to Y_{[t]} \| Z_{[t]}) \triangleq \sum_{i = 1}^{t}  I(X_{[i]}; Y_i | Y_{[i-1]}, Z_{[i]}). \label{eq:Idircond}
\end{equation}
\begin{lemma}[{\cite[(3.14)--(3.16)]{kramer1998PhD}}]
 Directed information chain rules:
\begin{align}
I((X_{[t]}, Y_{[t]}) \to Z_{[t]}) =&~ I(X_{[t]} \to Z_{[t]}) \notag\\
&+ I(Y_{[t]} \to Z_{[t]} \| X_{[t]}) \label{eq:chain1},\\
I( X_{[t]} \to (Y_{[t]}, Z_{[t]})) =&~ I(X_{[t]} \to Y_{[t]} \| \mathcal D Z_{[t]}) \notag\\
&+ I(X_{[t]} \to Z_{[t]} \| Y_{[t]}) \label{eq:chain2}.
\end{align}
\end{lemma}

Fix an $(M_{[t]}^{[K]}, d)$ code in \defnref{defn:ceocode}. Denote by $B_i^k \in [M_i^k]$ the codeword sent by $k$-th encoder at time $i$. Since the codewords satisfy the sum rate constraint \eqref{eq:sumrateconstraint}, 
\begin{align}
&~ n t  R \geq \sum_{k = 1}^K H(B_{[t]}^k) \\
 &\geq H \left(B_{[t]}^{[K]}\right) \label{eq:jointind} \\
 &\geq I \left( Y_{[t]}^{[K]} \to  B_{[t]}^{[K]} \right) \label{eq:dirent} \\
 &\geq \inf_{ \substack
 { P_{B_{[t]}^{[K]} \| Y_{[t]}^{[K]}} = \prod_{k = 1}^K P_{B_{[t]}^k \| Y_{[t]}^k}, \\
 P_{ \hat X_{[t]}^{[K]} \| B_{[t]}^{[K]}}  \colon \text{\eqref{eq:d} holds}} } I \left( Y_{[t]}^{[K]} \to  B_{[t]}^{[K]} \right), \label{eq:multiletter}
\end{align}
where \eqref{eq:jointind} holds because the joint entropy is upper-bounded by the sum of individual entropies, and \eqref{eq:dirent} holds because the mutual information is upper-bounded by the entropy. 
Note that \eqref{eq:multiletter} is the $n$-letter version of \eqref{eq:bti}. 

\detail{Since at time $i$, the decoder's knowledge about past observations is limited to $B_{[i-1]}^{[K]}$, in \eqref{eq:multiletter}, we may restrict our attention to encoding policies that satisfy
\begin{equation}
 P_{B_i^k | B_{[i-1]}^k, Y_{[i]}^k} =  P_{B_i^k | B_{[i-1]}^k, Y_{i}^k}.\label{eq:encpast}
\end{equation}
Indeed, given any $P_{B_{[t]}^{[K]} \| Y_{[t]}^{[K]}}$, we may construct a new policy by replacing $P_{B_i^k | B_{[i-1]}^k, Y_{[i]}^k}$ by $P_{B_i^k | B_{[i-1]}^k, Y_{i}^k}$. This will not change the joint distribution $P_{\hat X_{[t]}^{[K]}, B_{[t]}^{[K]}}$ and thus will not change \eqref{eq:d}, and it can only decrease the mutual information since $I (Y_{[i]}^{[K]}; B_i^{[K]} | B_{[i-1]}^{[K]}) \geq I (Y_{i}^{[K]}; B_i^{[K]} | B_{[i-1]}^{[K]})$.}

We proceed to apply a standard single-letterization argument to \eqref{eq:multiletter}. For an $n$-dimensional vector $Y_{i}^k$, we denote by $Y_{i}^k(j)$ its $j$-th component; for sets $\mathcal K \subseteq [K]$ and $\mathcal I \subseteq [n]$, we denote by $Y_{i}^{\mathcal K}(\mathcal I)$ the components of the vectors $\left( Y_{i}^{k} \colon k \in \mathcal K \right)$ indexed by $\mathcal I$. 

We introduce auxiliary random objects
\begin{align}
U_i^k(j)= \left(B_i^k, Y_i^{[K]}([j-1]) \right), \quad j \in [n] \label{eq:singleletterU}
\end{align} 

 The directed mutual information in the right side of \eqref{eq:multiletter} can be rewritten in terms of $U_i^{[K]}$ and bounded as follows. 
\begin{align}
&~ I \left( Y_{[t]}^{[K]} \to  B_{[t]}^{[K]} \right) \notag\\
=&~ \sum_{j = 1}^n I \left( Y_{[t]}^{[K]}(j) \to  B_{[t]}^{[K]}  \|  Y_{[t]}^{[K]}([j-1])\right) \label{eq:singlelettera} \\
=&~ \sum_{j = 1}^n I \left( Y_{[t]}^{[K]}(j) \to  \left( B_{[t]}^{[K]},  Y_{[t]}^{[K]}\left([j-1]\right) \right) \right) \notag\\
& - I \left( Y_{[t]}^{[K]}(j) \to  Y_{[t]}^{[K]}([j-1]) \| \D B_{[t]}^{[K]} \right) \label{eq:singleletterb} \\
=&~ \sum_{j = 1}^n I \left( Y_{[t]}^{[K]}(j) \to  \left( B_{[t]}^{[K]},  Y_{[t]}^{[K]}\left([j-1]\right) \right) \right)  \label{eq:singleletterc}\\
=&~ \sum_{j = 1}^n I \left( Y_{[t]}^{[K]}(j) \to  U_{[t]}^{[K]} (j)\right)  \label{eq:singleletterd} \\
\geq&~ \min_{\substack{ d_j, j \in [n] \colon \\
\sum d_j \leq n d  } } ~t \sum_{j = 1}^n \underline{\mathsf R}_{t\,\mathrm{CEO}}(d_j)  \label{eq:singlelettere}\\
\geq &~ n t\, \underline{\mathsf R}_{t\,\mathrm{CEO}}(d) \label{eq:singleletterf}
\end{align}
where \eqref{eq:singlelettera} is by the chain rule of mutual information; \eqref{eq:singleletterb} is by the chain rule of directed information \eqref{eq:chain2}; \eqref{eq:singleletterc} holds because $P_{B_{[t]}^{[K]} | Y_{[t]}^{[K]}} =  P_{B_{[t]}^{[K]} \| Y_{[t]}^{[K]}} $ is a causal kernel, which means that $P_{Y_{[t]}^{[K]} \| \D B_{[t]}^{[K]}  } = P_{Y_{[t]}^{[K]}}$, hence conditioning on $\D B_{[t]}^{[K]}$ in \eqref{eq:singleletterb} can be eliminated, and the resulting directed information is zero because different components of the vector $Y_i^k$ are independent due to \eqref{eq:singlelettersource}, \eqref{eq:singleletterchannel}; \eqref{eq:singleletterd} is by substituting \eqref{eq:singleletterU}; \eqref{eq:singlelettere} holds because $U_i^k(j)$ \eqref{eq:singleletterU}  satisfies 
$
P_{U_{[t]}^{k}(j) \| Y_{[t]}^{[K]}(j)} = P_{U_{[t]}^{k}(j) \| Y_{[t]}^{k}(j)}
$, the distortion measure is separable and \eqref{eq:singlelettersource}, \eqref{eq:singleletterchannel} hold; and \eqref{eq:singleletterf} is by the convexity of $\underline{\mathsf R}_{t\,\mathrm{CEO}}(d)$ as a function of $d$. \hfill $\qed$
\detail{
\begin{align}
 &~ P_{U_{i}^{k}(j) | Y_{[i]}^{[K]}(j), U_{[i-1]}^{k}(j)}  \notag\\
 =&~  P_{B_i^k, Y_i^{[K]}([j-1]) | B_{[i-1]}^k, Y_{[i-1]}^{[K]}([j-1]), Y_{[i]}^{[K]}(j)} \\
 =&~P_{B_i^k | B_{[i-1]}^k, Y_{[i]}^{[K]}([j])} \notag\\
 \cdot&~ P_{Y_i^{[K]}([j-1]) | Y_{[i-1]}^{[K]}([j-1]), Y_{[i]}^{[K]}(j),B_{[i-1]}^k}
\end{align} }
\subsection{\thmref{thm:cg}: proof of achievability}
\label{sec:a}
To show that \eqref{eq:cdp} is achievable in the asymptotics $n \to \infty$, we first show a nonasymptotic bound. Then, via an asymptotic analysis of the bound, we derive an extension of the Berger-Tung inner bound  \cite{berger1978multi,tung1978multiterminal} to the setting with inter-block memory. 

Before we present our nonasymptotic achievability bound in \thmref{thm:bt} below, we prepare some notation. 

For a fixed conditional distribution $P_{U_i^k  Y_{[i]}^k | U_{[i-1]}^k }$, denote the conditional information density 
\begin{align}
 \imath\left(y_{[i]}^k; u_i^k | u_{[i-1]}^k \right)  \triangleq \log \frac{d P_{U_i^k | Y_{[i]}^k, U_{[i-1]}^k }\left(u_i^k | y_{[i]}^k, u_{[i-1]}^k \right)}{d P_{U_i^k | U_{[i-1]}^k }\left(u_i^k | u_{[i-1]}^k \right)}. 
\end{align}

For a fixed joint distribution $P_{U_{[i]}^{[K]}}$, denote the relative conditional information densities
\begin{align}
\!\!\!\!\!\! \jmath^{k} \left(u_{[i]}^{[K]} \right) \triangleq   \log  \frac{dP_{U_{i}^{k} | U_{i}^{[k-1]} U_{[i-1]}^{[K]} } \!\! \left(u_i^{k}  \mid  u_i^{[k-1]} u_{[i-1]}^{[K]}  \right) }{  d P_{U_{i}^k | U_{[i-1]}^{k} } \left( u_i^k \  \mid  u_{[i-1]}^{k}   \right)  } .\!\!\!\!\! \label{eq:jsubset}
\end{align}
For a permutation $\pi \colon [K] \mapsto [K]$, we denote the ordered set
\begin{align}
 \pi(\mathcal K) \triangleq \left\{ \pi(k) \colon k \in \mathcal K\right\}.
\end{align}

\begin{thm}[nonasymptotic Berger-Tung inner bound with inter-block memory]
Fix $P_{Y_{[t]}^{[K]}}$ and parameters $M_{[t]}^{[K]}, d_{[t]}^{[K]}, \epsilon$. 
 For any scalars $\alpha_i^k, \beta_i^k$, any integers $L_i^k \geq M_i^k$, $i \in [t]$, $k \in [K]$, any causal kernels $P_{U_{[t]}^{[K]} \| Y_{[t]}^{[K]}} = \prod_{k = 1}^K P_{U_{[t]}^k \| Y_{[t]}^k}$ and $P_{ \hat X_{[t]}^{[K]} \| U_{[t]}^{[K]}}$, and any permutation $\pi \colon [K] \mapsto [K]$, there exists an $(M_{[t]}^{[K]}, d_{[t]}, \epsilon)$ excess distortion CEO code with inter-block memory such that 
\begin{align}
 \epsilon 
 &\leq \Prob{ \mathcal E}  + \gamma, \label{eq:btbound}
 \end{align}
 where event $\mathcal E$ is given by
 \begin{align}
 \mathcal E \triangleq & \bigcup_{i = 1}^t  \left\{ \sd \left(X_i, \hat X_i^k \right) > d_i \right\}\label{eq:event}\\
&\bigcup_{i = 1}^t \bigcup_{k = 1}^K \left\{   \imath\left(Y_{[i]}^k; U_i^k | U_{[i-1]}^k \right) > \log L_i^k - \alpha_i^k \right\} \notag \\
& \bigcup_{i = 1}^t \bigcup_{k = 1}^K \left\{   \jmath^{\pi(k)} \left(u_{[i]}^{\pi([K])} \right)   <  \log \frac{L_i^{\pi(k)}}{M_i^{\pi(k)}}  + \beta_i^{\pi(k)}\right\} ,  \notag 
\end{align}
and constant $\gamma$ is given by
\begin{align}
&~ \gamma \triangleq  1 - \label{eq:gamma}\\
&~ \frac {1}{ \prod_{i = 1}^t  \left[ \sum_{\mathcal K \subseteq K } \exp(-\sum_{k \in \mathcal K} \beta_i^{k})  \right] \prod_{k = 1}^K  \left[ 1 + \exp(-\alpha_i^k) \right]  }.\notag
\end{align}
\label{thm:bt}
\end{thm}
\begin{proof}[Proof sketch]
We employ the achievability proof technique developed by Yassaee et al. \cite{yassaee2013technique,yassaee2013techniqueArxiv} that uses a stochastic likelihood coder (SLC) to perform encoding operations. An SLC makes a randomized decision that coincides with high probability with the choice that a maximum likelihood (ML) coder would make (in fact, the error probability of the SLC exceeds by at most a factor of 2 the error probability of the ML coder \cite[Th. 7]{liu2017alpha}). We view the horizon-$t$ causal coding problem as a multiterminal coding problem in which at each step coded side information from past steps is available, and we define the SLC based on the auxiliary transition probability kernel $P_{U_i^k | Y_{[i]}^k U_{[i-1]}^k }$ (see \eqref{eq:btenc} in Appendix~\ref{apx:btnonas}) that is also used to generate random codebooks.

While \cite[Th.~6]{yassaee2013techniqueArxiv} shows a sharp nonasymptotic bound for the classical distributed source coding problem with $K=2$ terminals, the  decoder employed there does not extend to the case $K > 2$. In  \eqref{eq:btdec} in Appendix~\ref{apx:btnonas}, we propose a novel decoder that falls into the class of generalized likelihood decoders (GLD) conceptualized by Merhav \cite[eq. (4)]{merhav2017gld} and that uses an auxiliary indicator function $\mathsf g\left(u_{[i]}^{[K]}\right)$ \eqref{eq:gld}. With our GLD we are able to recover the full Berger-Tung region (\eqref{eq:btreg}, below) for any $K$.  One can view the set of outcomes $u_{[i]}^{[K]}$ for which $\mathsf g\left(u_{[i]}^{[K]}\right) = 1$ as a jointly typical set. That set depends on the choice of $\pi$ and thus on the particular rate point that the code is operating at. Checking for membership in that set involves $K$ threshold tests. In contrast, the jointly typical set defined by Oohama \cite[eq. (46)]{oohama1998ceo} involves $2^K - 1$ threshold tests, one for each nonempty subset of $[K]$.

Full details are given in Appendix~\ref{apx:btnonas}.
\end{proof}

\begin{thm}[Berger-Tung inner bound with inter-block memory]
 Under the assumptions of \thmref{thm:cg}, the rate-distortion tuple $(R^{[K]}, d)$ is asymptotically achievable at time horizon $t$ if for some 
 single-letter causal kernels $ P_{\mat U_{[t]}^{[K]} \| {\mat Y}_{[t]}^{[K]}}$, $P_{ \hat {\mathsf X}_{[t]}^{[K]} \| \mathsf U_{[t]}^{[K]}}$ satisfying \eqref{eq:sepenc1}, \eqref{eq:dt1} and some permutation $\pi \colon [K] \mapsto [K]$, it holds for all $k \in [K]$
  \begin{align}
R^{\pi(k)} > \frac 1 t I \left({\mat Y}_{[t]}^{\pi(k)} \to  \mat U_{[t]}^{\pi(k)}  \| \mat U_{[t]}^{\pi([k-1])}, \D \mat  U_{[t]}^{[K]}  \right). 
\label{eq:Rbt}
\end{align}
\label{thm:btas}
\end{thm}

\begin{proof}
 Appendix~\ref{sec:btas}.
\end{proof}
\thmref{thm:btas} implies that the sum rate 
\begin{align}
\sum_{k = 1}^K R^k > \frac 1 t I \left(  {\mat Y}^{[K]}_{[t]} \to  \mat U^{[K]}_{[t]} \right) \label{eq:dssumrate}
\end{align}
is achievable. Indeed, summing \eqref{eq:Rbt} over $k$ and using $\mathsf U_i^k - \left( \mathsf Y_{[i]}^k, \mathsf U_{[i-1]}^k \right) - \mathsf U_{[i]}^{[K] \backslash \{k\}}$ leads to \eqref{eq:dssumrate}. Therefore, the sum rate in \eqref{eq:bto} is achievable. \hfill \qed

\subsection{Remarks}
\label{sec:remarks}
We conclude \secref{sec:dir} with a set of remarks. 
\begin{enumerate}[1.]
\item  
Theorems~\ref{thm:bt} and ~\ref{thm:btas} are easily extended to distributed source coding with inter-block memory, where the goal is to separately compress (and jointly decompress) $K$ processes $\{Y_i^k\}$ under the individual distortion constraints 
\begin{align}
\frac 1 t \sum_{i = 1}^t \E{\sd^k(Y_i^k, \hat Y_i^k)} \leq d^k, \quad k \in [K].
\end{align}
 \thmref{thm:bt} continues to hold with $\sd \left(X_i, \hat X_i^k \right) > d_i$ in \eqref{eq:event} replaced by $\sd^k \left(Y_i^k, \hat Y_i^k \right) > d_i^k$. Consequently, \thmref{thm:btas} also continues to hold, replacing the constraint in \eqref{eq:dt1} by
\begin{align}
& \frac 1 {t} \sum_{i = 1}^t \E{\sd^k(\mat Y_i^k, \hat {\mat Y}_i^k) } \leq d^k, \quad k \in [K]. \label{eq:dk}
\end{align} 

\item Case $t = 1$ corresponds to the classical CEO / distributed source coding problems. The region in \eqref{eq:Rbt} simplifies to 
\begin{align}
\!\!\! R^{\pi(k)} > &~I(\sY^{\pi(k)}; \sU^{\pi(k)} | \sU^{\pi([k - 1])}), \notag\\
   &~\forall k \in [K], \forall \text{\,permutation $\pi \colon [K] \mapsto [K]$}. \label{eq:btperm}
\end{align}
The multiterminal Berger-Tung inner region is usually (e.g. \cite[Def. 7]{courtade2014multiterminal}, \cite[eq. (2)]{prabhakaran2004ceo}) specified as 
\begin{align}
\sum_{k \in \mathcal A} R^k > I(\sY^{\mathcal A}; \sU^{\mathcal A} | \sU^{\mathcal A^c}), \quad \forall \mathcal A \subseteq [K]. \label{eq:btreg}  
\end{align}
These characterizations are equivalent (Appendix~\ref{apx:bt}). 

\item While the sum rate bound in \eqref{eq:dssumrate} is the same regardless of the choice of permutation $\pi$, different $\pi$'s in \eqref{eq:Rbt} correspond to different orders in which the chain rule of mutual information can be applied, and are needed to specify the full achievable region of rates and distortions. 

\item We chose to omit the time-sharing random variable in \thmref{thm:btas} for simplicity of presentation. It can be introduced in \eqref{eq:Rbt} using the standard time sharing argument \cite[Ch. 4.4]{el2011network}. 
\end{enumerate}

\section{Gaussian rate-distortion function }
\label{sec:rd}

\subsection{Problem setup}
This section focuses on the scenario of the Gauss-Markov source in \eqref{eq:xi} observed through the Gaussian channels in \eqref{eq:yik} under squared error distortion \eqref{eq:MSEdef}.
Given an encoding policy in \defnref{defn:ceocode}, the optimal decoding policy $P_{ \hat X_{[t]}^{[K]} \| B_{[t]}^{[K]}}$ that achieves the minimum expected squared error is 
\begin{align}
 \hat X_i \triangleq \E{X_i | B_{[i]}^{[K]}}. 
 \label{eq:Xihat}
\end{align}

For simplicity we focus on the infinite time-horizon limit. 
\begin{align}
 R_{\mathrm{CEO}}(d) &\triangleq \limsup_{t \to \infty} R_{t\,\mathrm{CEO}}(d). \label{eq:Rceoinf}
\end{align}
In other words,  the causal CEO rate-distortion function $R_{\mathrm{CEO}}(d)$ is the infimum of $R$'s such that $\forall \gamma > 0$, $ \exists t_0 \geq 0$ such that $\forall t \geq t_0$,  $\exists n_0 \in \mathbb N$ such that $\forall n \geq n_0$, an $\left(M_{[t]}^{[K]}, d + \gamma \right)$ average distortion CEO code with inter-block memory exists with $M_{[t]}^{[K]}$ satisfying \eqref{eq:sumlogMik} and \eqref{eq:sumrateconstraint}.

Taking the limit $t \to \infty$ simplifies the solution of many minimal directed mutual information problems (\!\!\cite[Th.~9]{kostina2016ratecost}, \cite[Th.~6, Th.~7]{kostina2018SIallerton}, \cite[Th.~1]{tanaka2015stationary}, \cite[Th.~2]{guo2019wienersampling}, \cite[Th.~1]{sabag2021feedbackISIT}) by eliminating the transient effects due to the starting location $X_1$ of the process $\{X_i\}$ that is being transmitted. In this steady state regime, the optimal rate allocation across time is uniform (i.e., $\log M_1^k = \ldots = \log M_t^k$ in \eqref{eq:sumlogMik}). Furthermore, $R_{t\,\mathrm{CEO}}(d)$ approaches its steady-state value \eqref{eq:Rceoinf} as $\bigo{\frac 1 t}$ (this is a consequence of \cite[eq. (83)-(85), (92)]{kostina2018SIallerton} and \eqref{eq:cc}, \eqref{eq:Rnoiseless}, \eqref{eq:Rside} below). 

In \secref{sec:gceo}, we present the Gaussian rate-distortion function as a convex optimization problem over $K$ parameters (\thmref{thm:causalceo}), which reduces to an explicit formula in the identical-channels case (\corref{cor:sym}). These results are obtained by showing that the inner bound in \thmref{thm:cg} is tight in the Gaussian case and by evaluating the corresponding minimal directed mutual information. In \secref{sec:est}, we give auxiliary estimation lemmas that are useful in the proof of \thmref{thm:causalceo}. We give the proof of \thmref{thm:causalceo} in \secref{sec:proofmain}.

\emph{Notation:} For a random process $\{\sX_i\}$ on $\mathbb R$, its stationary variance (can be $+\infty$) is denoted by
\begin{align}
 \sigma_\sX^2 \triangleq \limsup_{i \to \infty} \E{ \sX_i^2}.
\end{align}
The minimum mean squared error (MMSE) in the estimation of $\sX_i$ from $\sY_{[i]}^{[K]}$ is denoted by 
\begin{align}
&~ \sigma_{\sX_i | \sY_{[i]}^{[K]}}^2 \triangleq \E{\left( \sX - \E{\sX_i| \sY_{[i]}^{[K]}}\right)^{2}},
\end{align}
and the steady-state causal MMSE by
\begin{align}
 \sigma_{\sX\| \sY^{[K]}}^2 \triangleq \limsup_{i \to \infty} \sigma_{\sX_i | \sY_{[i]}^{[K]}}^2.
\end{align}

\subsection{Gaussian rate-distortion function}
\label{sec:gceo}

In \thmref{thm:causalceo}, the Gaussian rate-distortion function is expressed as a convex optimization problem over parameters $\{d_k\}_{k = 1}^K$ that determine the individual rates of the transmitters and that correspond to the MSE achievable at the decoder in the estimation of $\{X_i\}_{i = 1}^t$ provided that the codewords from $k$-th transmitter are decoded correctly. 
\begin{thm}[Gaussian rate-distortion function with inter-block memory]
 For all $\sigma_{\sX\|\sY^{[K]}}^2 <  d < \sigma_{\mathsf X}^2$, the causal CEO rate-distortion function \eqref{eq:Rceoinf} for the Gauss-Markov source in \eqref{eq:xi} observed through the Gaussian channels in \eqref{eq:yik} is given by 
\begin{align}
R_{\mathrm{CEO}}(d) &= \frac 1 2 \log \frac{ \bar d }{d} 
+ \min_{\{d_k\}_{k = 1}^K}  \sum_{k = 1}^K \frac{1}{2} \log \frac{\bar d_k - \sigma_{\sX\| \sY^k}^2}{d_k - \sigma_{\sX\|\sY^k}^2}  \frac {d_k}{ \bar d_k}, \label{eq:minsumc}
\end{align}
where
\begin{align}
  \bar d &\triangleq a^2 d + \sigma_V^2, \label{eq:dprev}\\
    \bar d_k &\triangleq a^2 d_k + \sigma_V^2, \label{eq:dkbar}
\end{align}
and the minimum is over $d_k$, $k \in [K]$, that satisfy 
\begin{align}
 \frac 1 d &\leq \frac 1 {\sigma_{\sX\|\sY^{[K]}}^2} - \sum_{k = 1}^K \left( \frac 1 {\sigma^2_{\sX\|\sY^k}} - \frac 1 {d_k} \right), \label{eq:con1c}\\
\sigma_{\sX \| \sY^k}^2 &\leq d_k \leq \sigma_\sX^2.  \label{eq:con2c}
\end{align}
\label{thm:causalceo}
\end{thm}
\begin{proof}
\secref{sec:proofmain}. 
\end{proof}

If the source is observed directly by one or more of the encoders,  say if $\sigma_{\sX\| \sY^{1}}^2 = 0$, then $d_1 = d$, $d_2 = \ldots = d_K = \sigma_{\mathsf X}^2$ is optimal, and \eqref{eq:minsumc} reduces to the causal rate-distortion function \cite[eq. (1.43)]{gorbunov1974prognostic} (and e.g. \cite{tatikonda2004stochastic}, \cite[Th. 3]{derpich2012uppercausal}, \cite[(64)]{kostina2016ratecost}), \cite[Th. 6]{kostina2018SIallerton}):  
\begin{align}
 R(d) &= \frac{1}{2} \log \frac{\bar d}{d}. \label{eq:noiseless}
\end{align}
The sum over $k \in [K]$ in \eqref{eq:minsumc}  is thus the penalty due to the encoders not observing the source directly and not communicating with each other.

If the observation channels satisfy
\begin{align}
\sigma_{\mathsf X\| \mathsf Y^1}^2 = \ldots = \sigma_{\mathsf X\| \mathsf Y^K}^2, \label{eq:sym}
\end{align}
 we can explicitly write the rate-distortion function $R_{\mathrm{CEO}}^{K-\textrm{sym}}(d)$ for this symmetrical scenario.  

\begin{cor}[Observation channels with the same SNR]
If, in the scenario of \thmref{thm:causalceo}, the observation channels satisfy \eqref{eq:sym}, the causal CEO rate-distortion function \eqref{eq:Rceoinf} is given by 
\begin{align}
R_{\mathrm{CEO}}^{K-\textrm{sym}}(d) &= \frac 1 2 \log \frac{\bar d}{d} \label{eq:minsumcsym}
+ \frac{K}{2} \log \frac{\bar d_1 - \sigma_{\sX\| \sY^1}^2}{d_1 - \sigma_{\sX\|\sY^1}^2}  \frac {d_1}{ \bar d_1}, 
\end{align}
where $d_1$ satisfies
\begin{align}
 \frac 1 d &= \frac 1 {\sigma_{\sX\|\sY^{[K]}}^2} -  \frac K {\sigma^2_{\sX\|\sY^1}} + \frac K {d_1}. \label{eq:d1sym}
 \end{align} 
\label{cor:sym}
\end{cor}
\begin{proof}
 It suffices to show that the minimum in \eqref{eq:minsumc} is attained by $d_1 = \ldots = d_K$. Since each of the terms in the sum in \eqref{eq:minsumc} is a convex function of $d_k$, \detail{To see this, since each term is a rate-distortion function in $\rho_k$ \eqref{eq:Rside}, it is convex and nonincreasing in $\rho_k$, and since $\rho_k$ is a concave function of $d_k$ \eqref{eq:rhokdk}, it follows that each term is also convex in $d_k$.} applying Jensen's inequality concludes the proof. 
\end{proof}

Let us think now of adding identical observers by letting $K \to \infty$ in \eqref{eq:sym}. Since $\sigma_{\sX\|\sY^{[K]}}^2 \to 0$, had the observers communicated with each other, they could have recovered the source exactly, and they could have operated at the sum rate \eqref{eq:noiseless} in the limit. As the following result demonstrates, $\lim_{K \to \infty} R_{\mathrm{CEO}}^{K-\textrm{sym}}(d)$ is actually strictly greater than \eqref{eq:noiseless}, thus a nonvanishing penalty due to separate encoding is present in this regime. See \secref{sec:loss} for a more thorough discussion on the loss due to separate encoding.

\begin{cor}[Many channels asymptotics]
In the scenario of \corref{cor:sym},
  \begin{align}
\lim_{K \to \infty} R_{\mathrm{CEO}}^{K-\mathrm{sym}}(d) = \frac{1}{2} \log \frac{\bar d}{d} + \frac 1 2 \frac { \frac 1 d - \frac 1 {\bar d}}{  \frac 1 {\sigma_{\sX\|\sY^1}^2} - \frac 1 {\sigma_{\mathsf X}^2}}  . \label{eq:symlimc}
\end{align}
\label{cor:largeK}
\end{cor}
\begin{proof}
By \lemref{lem:combo} in \secref{sec:est} below,
 \begin{align}
 \frac 1 {\sigma_{\sX\|\sY^{[K]}}^2}   &=  \frac K {\sigma_{\sX\|\sY^1}^2}  - \frac{K - 1}{\sigma_{\sX}^2}. \label{eq:comboestK}
\end{align}
Eliminating $d_1$ and $\sigma_{\sX\|\sY^{[K]}}^2$ from \eqref{eq:minsumcsym} using \eqref{eq:d1sym} and \eqref{eq:comboestK}, one readily verifies that
\begin{align}
\!\!\!\!\! R_{\mathrm{CEO}}^{K-\mathrm{sym}}(d) - \frac 1 2 \log \frac{\bar d}{d} &= \frac 1 {2} \frac { \frac 1 d - \frac 1 {\bar d}} {  \frac 1 {\sigma_{\sX\|\sY^1}^2} - \frac 1 {\sigma_\sX^2}}   + \bigo{\frac 1 {K}},
\end{align}
and \eqref{eq:symlimc} follows.
\end{proof}

\corref{cor:largeK} extends the result of Oohama \cite[Cor. 1]{oohama1998ceo} to the compression with inter-block memory, and coincides with it if $a = 0$.


Considering the scenario where the encoders and the decoder do not memorize past observations or codewords, we may invoke the results on the classical Gaussian CEO problem in \cite{chen2004upper,prabhakaran2004ceo} to express the minimum achievable sum rate as
\begin{align}
R_{\mathrm{CEO}}^{\textrm{no memory}}(d) &= \frac 1 2 \log \frac{\sigma_{\sX}^2}{d} \notag\\
&+ \min_{\{d_k\}_{k = 1}^K}  \sum_{k = 1}^K \frac{1}{2} \log \frac{\sigma_{\sX}^2 - \sigma_{\sX | \sY^k}^2}{d_k - \sigma_{\sX |\sY^k}^2}  \frac {d_k}{ \sigma_{\sX}^2} \label{eq:minsumcnomem},
\end{align}
 where the minimum is over
\begin{align}
 & \frac 1 d \leq \frac 1 {\sigma_{\sX|\sY^{[K]}}^2} - \sum_{k = 1}^K \left( \frac 1 {\sigma^2_{\sX |\sY^k}} - \frac 1 {d_k} \right), \label{eq:con1cnomem}\\
& \sigma_{\sX | \sY^k}^2 \leq d_k  \leq \sigma_\sX^2.  \label{eq:con2cnomem}
\end{align} 
Here $\sigma_{\sX | \sY^k}^2 \triangleq \lim_{i \to \infty} \sigma_{\sX_i | \sY_i^k}$  and $\sigma_{\sX | \sY^{[K]}}^2 \triangleq \lim_{i \to \infty} \sigma_{\sX_i | \sY_i^{[K]}}^2$  denote the stationary MMSE achievable in the estimation of $\sX_i$ from $\sY_i^k$ and $\sY_i^{[K]}$ respectively, i.e., without memory of the past.  

If $a = 0$, the observed process \eqref{eq:xi} becomes a stationary memoryless Gaussian process, the predictive MMSEs reduce to the variance of $\sX_i$: $\bar d  = \bar d_k = \sigma_\sX^2 = \sigma_{\mathsf V}^2$; similarly, $\sigma_{\sX | \sY^k}^2 = \sigma_{\sX \| \sY^k}^2$ and $\sigma_{\sX | \sY^{[K]}}^2 = \sigma_{\sX \| \sY^{[K]}}^2$,
and the result of \thmref{thm:causalceo} coincides with the classical Gaussian CEO rate-distortion function \eqref{eq:minsumcnomem}. This shows that if the source is memoryless, asymptotically there is no benefit in keeping the memory of previously encoded estimates as permitted by \defnref{defn:ceocode}.   Classical codes that forget the past after encoding the current block of length $n$ perform just as well. 

If $|a| > 1$, the benefit due to memory is infinite: indeed, since the source is unstable, $\sigma_{\sX}^2 = \infty$,  while $\bar d < \infty$. If $|a| < 1$, that benefit is finite and is characterized by the discrepancy between  the stationary variance $\sigma_\sX^2 = \frac{\sigma_{\mathsf V}^2}{1 - a^2}$ of the process $\{\sX_i \}_{i = 1}^\infty$ and the steady-state predictive MMSE $\bar d < \sigma_\sX^2$, as well as that between $\sigma_{\sX | \sY^k}^2$ and $\sigma_{\sX \| \sY^k}^2$.

\subsection{MMSE estimation lemmas}
\label{sec:est}
We record two elementary estimation lemmas that will be instrumental in the proof of \thmref{thm:causalceo}. 
\begin{lemma}
Let $X \sim \Gauss{0, \sigma_{X}^2}$, $W \sim \Gauss{0, \sigma_{W}^2}$, $W \perp X$, and let   
\begin{align}
Y = X + W.
\end{align}
Then, 
\begin{align}
\sigma_{X | Y}^2 
&= \sigma_{X}^2 \left( 1 - \frac{\sigma_{X}^2}{\sigma_{Y}^2} \right). \label{eq:cond1}
\end{align}
\label{lem:back}
\end{lemma}
\begin{proof}
Appendix~\ref{apx:gest}.
\end{proof}
\detail{
\begin{remark}
Dropping the assumptions of Gaussianity but keeping those of uncorrelatedness in Lemmas \ref{lem:Gest}--\ref{lem:back}, relations \eqref{eq:msek} and \eqref{eq:cond1} continue to hold replacing the normalized conditional variances $\sigma_{X|Y_{[K]}}^2$ and $\sigma_{X|Y}^2$ with the MMSEs achieved by the optimal \emph{linear} estimator. 
\label{rem:nogauss}
\end{remark}
}

\begin{lemma}
Let $\bar X_k$ and $W_k^\prime$ be Gaussian random variables, $ \left\{\bar X_k \right\}_{k=1}^K \perp \{W_j^\prime\}_{j=1}^K$,  such that
$W_k^\prime \perp W_j^\prime$, $j \neq k$, and   
\begin{align}
 X = \bar X_k + W_k^\prime. \label{eq:backx}
 \end{align}
 Then, the MMSE estimate and the estimation error $\sigma_{W^\prime}^2 \triangleq \sigma_{X | \bar {X}_{[K]}}^2$ of $X$ given the vector $\bar X_{[K]}$ satisfy
\begin{align}
 \E{X| \bar X_{[K]}} &=  \sum_{k = 1}^K   \frac{\sigma_{W^\prime}^2}  {\sigma_{W_k^\prime}^2} \bar X_k \label{eq:combo},\\ 
 \frac 1 {\sigma_{W^\prime}^2} &=  \sum_{k =1}^K \frac 1 {\sigma_{W_k^\prime}^2 } - \frac{K-1}{\sigma_{X}^2}. 
 \label{eq:msecombo}
\end{align}
\label{lem:combo}
\end{lemma}
\begin{proof}
Appendix~\ref{apx:gest}.
\end{proof}

\lemref{lem:combo} converts the ``forward channels'' from $X$ to observations $Y_k$
\begin{align}
Y_k = X + W_k,~ k = 1, \ldots, K,
\end{align}
where $W_k \sim \Gauss{0, \sigma_{W_k}^2 \mathbf I}$ ($\mathbf I$ denotes the identity matrix), $W_k \perp W_j$, $j \neq k$,
  into ``backward channels'' from estimates $\bar X_k$ to $X$ \eqref{eq:backx}. While both representations are equivalent, \eqref{eq:backx} is more convenient to work with. Backward channel representations find a widespread use in rate-distortion theory \cite{berger1971rate}.

\subsection{Proof of \thmref{thm:causalceo}: converse}
\label{sec:proofmain}

\subsubsection{Proof overview}
We evaluate the $n$-letter converse bound \eqref{eq:multiletter}. 
We break up the minimal directed mutual information problem in \eqref{eq:multiletter} into subproblems, and we use the tools we developed in \cite{kostina2018SIallerton} to evaluate the causal rate-distortion functions for each subproblem. To link the parameters of the subproblems together to obtain the solution of the original problem, we extend the proof technique by Wang et al. \cite{wang2010sum}, developed for the case $t = 1$, to $t > 1$. Converting the ``forward channels'' from $X_{[t]}$ to observations $Y_{[t]}^k$ into the ``backward channels'' from MMSE estimates $\bar{X}_{[t]}^k$ to $X_{[t]}$ and applying the lemmas in \secref{sec:est} above are key to that extension. 

\subsubsection{Decoupling the problem into $K$ subproblems}
Recall the notation in \eqref{eq:Xbark}. 
We expand the right-hand side of \eqref{eq:multiletter}:
\begin{align}
\!\!\! &  \inf I \left(  Y_{[t]}^{[K]} \to   B_{[t]}^{[K]} \right) \notag\\
\!\!\!\geq & \inf  I \left( \bar { X}_{[t]}^{[K]} \to   B_{[t]}^{[K]} \right) \label{eq:dpbar} \\
\!\!\! =& \inf  I\left(  \left( X_{[t]},\bar { X}^{[K]}_{[t]} \right)  \to  B^{[K]}_{[t]} \right) \label{eq:ca} \\
\!\!\!=& \inf  \left\{  I \left( X_{[t]} \to  B^{[K]}_{[t]} \right) +  I \left( \bar { X}^{[K]}_{[t]} \to  B^{[K]}_{[t]} \|  X_{[t]} \right)  \right\} \label{eq:cb} \\
\!\!\! =& \inf  \left\{  I \left( X_{[t]} \to  B^{[K]}_{[t]} \right) +  \sum_{k = 1}^K   I \left( \bar { X}^k_{[t]} \to  B^k_{[t]}\|  X_{[t]} \right)   \right\}\label{eq:cc} 
 \end{align}
where 
\begin{itemize}
\item \eqref{eq:dpbar} holds by the chain rule \eqref{eq:chain1} using $I\left(\bar {X}_{[t]}^{[K]} \to B_{[t]}^{[K]} \| Y_{[t]}^{[K]} \right) = 0$. 
The infimum is over kernels $P_{B_{[t]}^{[K]} \| \bar { X}_{[t]}^{[K]}}$ satisfying both the separate encoding constraint
\begin{align}
  P_{B_{[t]}^{[K]} \| \bar {X}_{[t]}^{[K]}} &= \prod_{k = 1}^K P_{B_{[t]}^k \| \bar {X}_{[t]}^k}
  \label{eq:encpolicyj1}
\end{align}
and the distortion constraint
\begin{align}
& \frac 1 {nt}  \sum_{i = 1}^t \E{ \| X_i - \hat X_i\|^2} \leq d, \label{eq:d1letter}
\end{align} 
where $\hat X_i$ \eqref{eq:Xihat} is the MMSE estimate of $X_i$ given $B_{[i]}^{[K]}$;
\item  \eqref{eq:ca} is due to the chain rule of directed information~\eqref{eq:chain1}, and $I\left( X_{[t]} \to B_{[t]}^{[K]} \| \bar {X}_{[t]}^{[K]}\right) = 0$; 
\item \eqref{eq:cb} is by the chain rule of directed information \eqref{eq:chain1};
\item \eqref{eq:cc} is due to \eqref{eq:encpolicyj1}. 
\end{itemize}

\subsubsection{Using causal rate-distortion functions to evaluate the terms in \eqref{eq:cc}}
We lower-bound the first term in \eqref{eq:cc} using a classical result on the point-to-point causal Gaussian rate-distortion function \cite[eq. (1.43)]{gorbunov1974prognostic}\footnote{See also \cite[Th. 6]{kostina2018SIallerton}; while stated for the scalar Gaussian source, the same argument applies to $n$ parallel Gaussian sources of the same power, as is the case here; see \cite{tanaka2017semidefinite} for the general vector case.} as
\begin{align}
&~\lim_{t \to \infty} \inf_{ 
\substack{
\eqref{eq:encpolicyj1} \colon \\
\eqref{eq:d1letter} \text{ holds}}
} \frac{1}{t}  I  \left(X_{[t]} \to B^{[K]}_{[t]} \right) \notag\\
\geq &~
\lim_{t \to \infty} \inf_{ 
\substack{
P_{ \hat X^{[K]}_{[t]} \| X_{[t]}} \colon\\
\eqref{eq:d1letter} \text{ holds}}
} \frac{1}{t} I  \left(X_{[t]} \to \hat X^{[K]}_{[t]} \right)  \label{eq:Th2SI}\\
 =&~ \frac {n} 2 \log \frac{\bar d }{d}, \label{eq:Rnoiseless}
\end{align}
where $\bar d$ is uniquely determined by $d$ via \eqref{eq:dprev}. Furthermore, \eqref{eq:Rnoiseless} is achieved by the Gaussian kernel $P_{\hat X^\star_{[t]} \| X_{[t]}}$ such that
\begin{align}
X_i = \hat {X}_i^\star + Z_i^\prime, \quad Z_i^\prime \sim \mathcal N(0, d\, \mathbf I), \label{eq:Xstarback}
\end{align}
 $\{Z_i^\prime\}$ are i.i.d. and independent of $\{\hat {X}_i^\star\}$, and 
\begin{align}
d &=  \sigma_{\mat X\| \hat{\sX}^\star}^2 \label{eq:dviaXstar}\\
\bar d &=  \sigma_{\mat X\| \D \hat{\sX}^\star}^2. 
\end{align}

For each of the remaining $K$ terms in \eqref{eq:cc}, note that $\{\bar X_i^k\}$ is a Gauss-Markov process
\begin{align}
 \bar X_{i+1}^k = a \bar X_i^k + \bar {V}_i^k, \label{eq:Xibargm}
\end{align}
where $\bar {V}_i^k \sim \mathcal N \left( 0, \left(\sigma_{X_i^k | Y_{[i]}^k}^2 - \sigma_{X_i^k | Y_{[i-1]}^k}^2 \right) \mathbf I \right)$. 
The process $\{X_i\}$ can be expressed through  $\{\bar {X}_i^k\}$ as
\begin{align}
X_i = \bar {X}_i^k + W^{ k \, \prime}_i, \label{eq:Xibarplus}
\end{align} 
 where $W^{ k \, \prime}_i$ are independent, $W^{ k \, \prime}_i \sim \Gauss{0, \sigma_{\mathsf X_i | \mathsf Y_{[i]}^k }^2 \mathbf I}$, and $W^{ k \, \prime}_i \perp X_i^k$.  Thus, we may apply the result \cite[Th. 7]{kostina2018SIallerton} on the causal counterpart of Gaussian Wyner-Ziv rate-distortion function  to the process $\{\bar{X}^k_i \}$ \eqref{eq:Xibargm} with side information $\{X_i\}$ \eqref{eq:Xibarplus} to write (while stated for the scalar Gaussian source, the same argument applies to $n$ parallel Gaussian sources of the same power, as is the case here; see \cite{sabag2020minimal} for the general vector case)
\begin{align}
&~\lim_{t \to \infty} \inf_{ 
\substack{
P_{B^k_{[t]} \| \bar X_{[t]}^k } \colon \\
\frac 1 t \sum_{i = 1}^t \sigma_{\bar {X}^k_{i}  | X_{[i]},  B^{k}_{[i]}}^2  \leq \rho_k }
}  \frac 1 t I \left( \bar {X}^k_{[t]} \to B^k_{[t]}\| X_{[t]} \right) \\
 =&~ \frac n 2 \log \frac{\bar \rho_k }{\rho_k}, \label{eq:Rside}
\end{align}
where $\bar \rho_k$ is uniquely determined by $\rho_k$ via
\begin{align}
\frac 1 {\bar \rho_k} = \frac{1}{\sigma_{\bar {\mathsf W}^{k \prime}}^2} + \frac 1 {a^2 \rho_k + \sigma_{\bar {\mathsf V}}^2}.
\end{align}
Furthermore, \eqref{eq:Rside} is attained by the Gaussian kernel $P_{B^{k \star} \|  \bar {X}^k}$
\begin{align}
B_{i}^{k \star} &= \bar {X}_i^k + Z_i^k, \quad Z_i \sim \mathcal N(0, \sigma_{\mat Z^k}^2 \mathbf I),
\label{eq:test}
\end{align}
 $\{Z_i\}$ are i.i.d. and independent of $\{\bar {X}_i^k\}$, and
\begin{align}
\rho_k &=  \sigma_{\bar {\mathsf X}^k \| \mathsf X, \mathsf B^{k\star} }^2,  \label{eq:rhoku} \\
\bar \rho_k &= \sigma_{\bar {\mathsf X}^k \| \mathsf X, \D \mathsf B^{k\star} }^2.
\end{align}
The variances $\sigma_{\mat Z^k}^2$ in \eqref{eq:test} are set to satisfy \eqref{eq:rhoku}. 
 
\subsubsection{Linking $\{\rho_k\}_{k = 1}^K$ to $d$}
 It remains to establish the connection between $\{\rho_k\}_{k = 1}^K$ \eqref{eq:rhoku} and $d$ \eqref{eq:dviaXstar}. 
 
 Setting $\hat X_i^\star$ in \eqref{eq:Xstarback} to 
 \begin{align}
 \hat X_i^\star \triangleq \E{X_i | B_{[i]}^{[K]\,\star}}
 \label{eq:Xihatstar}
\end{align}
 attains equality in \eqref{eq:Th2SI}, implying that the same Gaussian kernel \eqref{eq:test} simultaneously attains the infima of both terms in \eqref{eq:cc}. Thus, putting together \eqref{eq:cc}, \eqref{eq:Rnoiseless} and \eqref{eq:Rside}, we have 
\begin{align}
&~ R_{\mathrm{CEO}}(d) \geq \label{eq:agoal} \\
 &\inf_{ \substack{ \left\{ \sigma_{\bar {\mathsf X}^k \| \mathsf X, \mathsf U^{k\star} }^2 \right\}_{k = 1}^K \!\! \colon \\
\sigma_{\sX \| \mathsf U^{[K] \star} }^2 = d }
 } \!\!\! \left\{ \frac 1 2 \log \frac{\sigma_{\mat X\| \D \mat B^{[K] \star}}^2 }{\sigma_{\sX \| \mathsf B^{[K] \star} }^2} + \sum_{k = 1}^K \frac 1 2  \log \frac{\sigma_{\bar {\mat X}^k \| \mat X, \D {\mat B^\star}^k}^2 }{\sigma_{\bar {\mathsf X}^k \| \mathsf X, \mathsf B^{k\star} }^2 } \right\}
\notag
\end{align}
Invoking \lemref{lem:combo} with $X \leftarrow \mathsf X_i$, $\bar X_k \leftarrow \bar {\mathsf X}_i^k$, $W_k^\prime \leftarrow \mathsf W^{ k \, \prime}_i$, we express
\begin{align}
 \bar {\mathsf X}_i &\triangleq \E{\mat X_i | \mat Y_{[i]}^{[K]}} \label{eq:Xibar1let}\\
 &= \sum_{k = 1}^K \frac{\sigma_{\mathsf X_{i} |  \mathsf Y^{[K]}_{[i]}}^2}  {\sigma_{\mathsf X_i | \mathsf Y^k_{[i]}}^2} \bar {\mathsf X}_{i}^k, 
\end{align}
which implies in particular
\begin{align}
\!\!\! \E{\bar {\mathsf X}_i| \mathsf X_{[i]}, \mathsf B_{[i]}^{[K] \star} } &= \sum_{k = 1}^K \frac{\sigma_{\mathsf X_i |  \mathsf Y_{[i]}^{[K]}}^2}  {\sigma_{\mathsf X_i |  \mathsf Y^k_{[i]}}^2}   \E{\bar {\mathsf X}_i^k | \mathsf X_{[i]}, \mathsf B_{[i]}^{[K] \star}} \\ 
 &= \sum_{k = 1}^K \frac{\sigma_{\mathsf X_i | \mathsf Y^{[K]}_{[i]}}^2}  {\sigma_{\mathsf X_i | \mathsf Y^k_{[i]}}^2}   \E{\bar {\mathsf X}_i^k | \mathsf X_{[i]}, \mathsf B_{[i]}^{k \star}}.
\end{align}
It follows that steady-state causal MMSE in estimating $\bar {\mathsf X}_i$ from $\mathsf X_{[i]}$ and $\mathsf B^{[K] \star}_{[i]}$ satisfies
\begin{align}
\sigma_{\bar {\mathsf X} \| \mathsf X, \mathsf B^{[K] \star} }^{2} =  \sum_{k = 1}^K \frac{\sigma_{\mathsf X \| \mathsf Y}^4}  {\sigma_{\mathsf X \| \mathsf Y^k}^4} \rho_k \label{eq:sumdk}.
\end{align}
Observe that
\begin{align}
\sigma_{\bar {\mathsf X}_i | \mathsf X_{[i]}, \mathsf B^{[K] \star}_{[i]} }^2 &=  \sigma_{\bar {\mathsf X}_i - \E{\bar {\mathsf X}_i | \mathsf X_{[i]}, \mathsf B^{[K] \star}_{[i]} } }^2  \\
&= \sigma_{\bar {\mathsf X}_i -  \mathsf X_i - \E{\bar {\mathsf X}_i - \mathsf X_i | \mathsf X_{[i]}, \mathsf B^{[K]^\star}_{[i]} } }^2\\
 &= \sigma_{\bar {\mathsf X}_i - \mathsf X_i - \E{\bar {\mathsf X}_i  - \mathsf X_i | \mathsf X_i - \hat {\mathsf X}_i^\star }}^2 \label{eq:domia}\\
 &= \sigma_{\mathsf X_i - \bar {\mathsf X}_i  | \mathsf X_i - \hat {\mathsf X}_i^\star}^2  \label{eq:domi},
 \end{align}
Now, we apply \lemref{lem:back} with $X \leftarrow \mathsf X_i - \bar {\mathsf X}_i$, $Y \leftarrow \mathsf X_i - \hat {\mathsf X}_i$,  $W \leftarrow \bar {\mathsf X}_i - \hat {\mathsf X}_i$ to establish
\begin{align}
\lim_{i \to \infty} {\sigma}_{\mathsf X_i - \bar {\mathsf X}_i  | \mathsf X_i - \hat {\mathsf X}_i}^2  &= \sigma_{\mathsf X\| \mathsf Y}^2 \left( 1 - \frac{\sigma_{\mathsf X\| \mathsf Y}^2}{d}\right),
 \label{eq:key}
\end{align}
which, together with \eqref{eq:sumdk} and \eqref{eq:domi}, means
\begin{align}
 \frac 1 d \leq \frac 1 {\sigma_{\mathsf X\|\mathsf Y}^2} - \sum_{k = 1}^K \frac {\rho_k}{\sigma_{\mathsf X\|\mathsf Y_k}^4}. \label{eq:drhok}
\end{align}
Also, note that 
\begin{align}
0 \leq \rho_k \leq \sigma_{\bar {\mathsf X}^k \| \mathsf X}^2. 
\label{eq:rhokbound}
\end{align}
We can now simplify the constraint set in the infimum in \eqref{eq:agoal}: the infimum is over $\{\rho_k\}_{k = 1}^K$ that satisfy \eqref{eq:drhok} and \eqref{eq:rhokbound}.

It remains to clarify how the form in \eqref{eq:minsumc}, \eqref{eq:con1c}, \eqref{eq:con2c}, parameterized in terms of
\begin{align}
d_k &\triangleq \sigma_{\mathsf X\| \mathsf B^{k \star}}^2
\end{align}
rather than $\rho_k$, is obtained. 
An application of \lemref{lem:back} with $X \leftarrow \mathsf X_i - \bar {\mathsf X}_i^k$, $Y \leftarrow \mathsf X_i - \hat {\mathsf X}_i^k$,  $W \leftarrow \bar {\mathsf X}_i^k - \hat {\mathsf X}_i^k$ leads to
\begin{align}
   \rho_{k} &=  \sigma_{\mathsf X\| \mathsf Y^k}^2 \left(1 - \frac{\sigma_{\mathsf X\|\mathsf Y^k}^2}{d_k} \right) \label{eq:rhokdk}.
\end{align}
Plugging \eqref{eq:rhokdk} into \eqref{eq:drhok} leads to  \eqref{eq:con1c}. Applying \lemref{lem:back} with $X \leftarrow \mathsf X_i - \bar {\mathsf X}_i^k$, $Y \leftarrow \mathsf X_i$,  $W \leftarrow \bar {\mathsf X}_i^k $, we express 
\begin{align}
 \sigma_{\bar {\mathsf X}^k \| \mathsf X}^2 &= \sigma_{\mathsf X\| \mathsf Y^k}^2 \left(1 - \frac{\sigma_{\mathsf X\|\mathsf Y^k}^2}{\sigma_{\mathsf X}^2} \right), \label{eq:rhofirst}
\end{align}
which, together with \eqref{eq:rhokdk}, implies the equivalence of \eqref{eq:rhokbound} and \eqref{eq:con2c}. Finally, applying \lemref{lem:back} with $X \leftarrow \mathsf X_i - \bar {\mathsf X}_i^k$, $Y \leftarrow \mathsf X_i - a \hat {\mathsf X}_{i-1}^k$,  $W \leftarrow \bar {\mathsf X}_i^k - a \hat {\mathsf X}_{i-1}^k$, we express

\begin{align}
  \bar \rho_k &=  \sigma_{\mathsf X\| \mathsf Y^k}^2 \left(1 - \frac{\sigma_{\mathsf X\| \mathsf Y^k}^2}{\bar d_k} \right) \label{eq:rhoprev}.
\end{align}
Plugging \eqref{eq:rhokdk}  and \eqref{eq:rhoprev} into \eqref{eq:agoal}, we conclude the equivalence of \eqref{eq:agoal} and \eqref{eq:minsumc}.  
\hspace*{\fill}$\qed$

\subsection{Proof of \thmref{thm:causalceo}: achievability}
We evaluate the Berger-Tung inner bound with inter-block memory \eqref{eq:bti}. In the proof of the converse, we lower-bounded the $n$-letter version of that bound, i.e., \eqref{eq:multiletter}, by computing the right-hand side of \eqref{eq:cc}. Thus, it suffices to show that equality holds in \eqref{eq:dpbar}. But this is easily verified by substituting the optimal kernel \eqref{eq:test} into the left side of~\eqref{eq:dpbar}. \hspace*{\fill}$\qed$

\section{Loss due to isolated observers}
\label{sec:loss}

\subsection{Overview}
In \secref{sec:loss}, we investigate how the rate-distortion function in \thmref{thm:causalceo} compares to what would have been achievable had the encoders communicated with each other. A tight upper bound on the rate loss due to separate encoding is presented in \secref{sec:lossmain} (\thmref{thm:loss}).  Its proof relies on an upper bound on $R_{\mathrm{CEO}}(d)$ presented in \secref{sec:wf} (\propref{prop:causalwater}). The proof of \thmref{thm:loss} in \secref{sec:proofloss} concludes the section.

\subsection{Loss due to isolated observers}
\label{sec:lossmain}

Unrestricted communication among the encoders is equivalent to having one encoder that sees all the observation processes $\left\{Y_{i}^{[K]}\right\}$. It is also equivalent to allowing joint encoding policies $P_{B_{[t]}^{[K]} \| Y_{[t]}^{[K]}}$ in lieu of independent encoding policies $\prod_{k = 1}^K P_{B_{[t]}^k \| Y_{[t]}^k}$ in \defnref{defn:ceocode}. 

The lossy compression setup in which the encoder has access only to a noise-corrupted version of the source has been referred to as ``remote'', ``indirect'', or ``noisy'' rate-distortion problem in the literature \cite{dobrushin1962addnoise,berger1971rate,witsenhausen1980indirect,kostina2016noisy}. The setting with causal coding was considered in \cite[Th. 5--8, Cor. 1]{kostina2016ratecost}. 

We denote the joint encoding counterpart of the operational fundamental limit $R_{\mathrm{CEO}}(d)$ \eqref{eq:Rceoinf}  by $R_{\mathrm{rm}}(d)$ (remote). 

The following result is a corollary to \thmref{thm:causalceo}. 
\begin{cor}[Remote rate-distortion function with inter-block memory]
 For all $\sigma_{\sX\|\sY^{[K]}}^2 <  d < \sigma_{\mathsf X}^2$, the rate-distortion function with joint encoding for the Gauss-Markov source in \eqref{eq:xi} observed through the Gaussian channels in \eqref{eq:yik} is given by 
\begin{align}
 R_{\mathrm{rm}}(d) &= \frac{1}{2} \log \frac{\bar d - \sigma_{\sX\| \sY^{[K]}}^2}{d - \sigma_{\sX\|\sY^{[K]}}^2}, \label{eq:noisy}
\end{align}
where $\bar d$ is defined in \eqref{eq:dprev}.
\label{cor:noisy}
\end{cor}
\begin{proof}
Examining its proof, it is easy to see that \thmref{thm:causalceo} continues to hold in the scenario with vector observations $\mathsf Y_i^k$ (that are still required to be jointly Gaussian with $\mathsf X_i$). In light of this fact, we view the joint encoding scenario as the CEO scenario with a single encoder that has access to all $K$ observations, and we see that \eqref{eq:minsumc} indeed reduces to \eqref{eq:noisy} in that case. 

Previously, the minimal mutual information problem leading to $R_{\mathrm{rm}}(d)$ was solved in \cite{kostina2016ratecost} in a different form using a different method; both forms are equivalent (Appendix~\ref{apx:noisy}). 
\end{proof}

The loss due to isolated encoders is bounded as follows. 
\begin{thm}[Loss due to isolated observers]
Consider the causal Gaussian CEO problem \eqref{eq:xi}, \eqref{eq:yik}. Assume that target distortion $d$ satisfies $\sigma_{\sX\|\sY^{[K]}}^2 <  d$ and
\begin{align}
\frac 1 d \geq \frac{1}{\sigma_{\sX\|\sY^{[K]}}^2} + \frac K {\sigma_\sX^2} -   \min_{k\in [K]} \frac K { \sigma_{\sX\| \sY^k}^2}.\label{eq:dsmallenough}
\end{align}
Then, the rate loss due to isolated observers is bounded as 
\begin{align}
R_{\mathrm{CEO}}(d) - R_{\mathrm{rm}}(d)  &\leq (K - 1) \left( R_{\mathrm{rm}}(d) - R(d) \right) \label{eq:lossc},
\end{align}
with equality if and only if $\sigma_{\sX\| \sY^k}^2$ are all the same, where $R(d)$ is given in \eqref{eq:noiseless} and $R_{\mathrm{rm}}(d)$ is given in \eqref{eq:noisy}.
\label{thm:loss}
\end{thm}
\begin{proof}
\secref{sec:proofloss}. 
\end{proof}
\thmref{thm:loss} parallels the corresponding result for the classical Gaussian CEO problem \cite[Cor. 1]{kostina2019ratelossITW}, and recovers it if $a = 0$. It is interesting that in both cases, the rate loss is bounded above by $K-1$ times the difference between the remote and the direct rate-distortion functions. In the case of identical observation channels, condition \eqref{eq:dsmallenough} reduces to $d \leq \sigma_{\mathsf X}^2$. The rate loss \eqref{eq:lossc} grows without bound in the high resolution regime $d \downarrow \sigma_{\sX\|\sY^{[K]}}^2$ and vanishes in the low resolution regime $d \uparrow \sigma_{\mathsf X}^2$. 
\detail{The rate loss \eqref{eq:lossc} grows without bound as $(K-1) \log \frac{1}{d -  \sigma_{\sX\|\sY^{[K]}}^2}$ in the high resolution regime $d \downarrow \sigma_{\sX\|\sY^{[K]}}^2$. The rate loss vanishes in the low resolution regime $d \uparrow \sigma_{\mathsf X}^2$, and in the stationary case, it vanishes linearly in $\sigma_{\mathsf X}^2 - d$.}

\subsection{A suboptimal waterfilling allocation}
\label{sec:wf}

We present an upper bound to $R_{\mathrm{CEO}}(d)$, which is obtained by waterfilling over $d_k$'s.
This parallels the corresponding result for the classical Gaussian CEO problem \cite[Cor. 1]{kostina2019ratelossITW}. Like  \cite{kostina2019ratelossITW}, we use waterfilling to obtain this result, but unlike the case $t = 1$ considered in \cite{kostina2019ratelossITW} where waterfilling is optimal  \cite{chen2004upper}, it is only suboptimal if $t > 1$ due to the memory of the past steps at the encoders and the decoder. This is unsurprising, as for the same reason waterfilling cannot be applied to solve the vector Gaussian rate-distortion problem for $t > 1$ \cite[Remark 2]{kostina2016ratecost}.

\begin{prop}[Suboptimal waterfilling rate allocation]
For all $\sigma_{\sX\|\sY^{[K]}}^2 <  d < \sigma_{\mathsf X}^2$, the causal CEO rate-distortion function for the Gauss-Markov source in \eqref{eq:xi} observed through the Gaussian channels in \eqref{eq:yik}  is upper-bounded as 
\begin{align}
R_{\mathrm{CEO}}(d) &\leq \frac 1 2 \log \frac{\bar d}{d}+  \sum_{k = 1}^K \frac{1}{2} \log \frac{\bar d_k - \sigma_{\sX\| \sY^k}^2}{d_k - \sigma_{\sX\|\sY^k}^2}  \frac {d_k}{ \bar d_k}  ,  \label{eq:minwater} 
\end{align}
where $d_k$, $k \in [K]$ satisfy
\begin{align}
\frac 1 {\sigma^2_{\sX\|\sY^k}} - \frac 1 {d_k} =  \min \left\{ \frac 1 \lambda,\, \frac 1 {\sigma_{\sX\|\sY^k}^2} - \frac 1 {\sigma_{\sX}^2}  \right\}, 
\label{eq:dkwater}
\end{align}
 $\lambda$ is the solution to
\begin{align}
\sum_{k = 1}^K  \min \left\{ \frac 1 \lambda,\, \frac 1 {\sigma_{\sX\|\sY^k}^2} - \frac 1 {\sigma_{\sX}^2}  \right\} 
= \frac 1 {\sigma_{\sX\|\sY^{[K]}}^2} - \frac 1 d , \label{eq:lambda}
\end{align}
and $\bar d$, $\bar d_k$ are defined in \eqref{eq:dprev}, \eqref{eq:dkbar} respectively. Inequality in \eqref{eq:minwater} holds with equality if all $\sigma_{\sX\|\sY^k}^2$ are equal. 

\label{prop:causalwater}
\end{prop}

\begin{proof}
We first check that the choice in \eqref{eq:dkwater} is feasible. Since the right side of \eqref{eq:dkwater} is lower-bounded by 0 and upper bounded by $\frac 1 {\sigma_{\sX\|\sY^k}^2} - \frac 1 {\sigma_{\sX}^2}$, \eqref{eq:con2c} is satisfied. Furthermore, substituting 
\eqref{eq:lambda} ensures that \eqref{eq:con1c} is satisfied with equality.  

To claim equality in the symmetrical case, it suffices to recall that in that case, the minimum in \eqref{eq:minsumc} is attained by $d_1 = \ldots = d_K$ (\corref{cor:sym}). 
\end{proof}

\subsection{Proof of \thmref{thm:loss}}
\label{sec:proofloss} 
Under the assumption \eqref{eq:dsmallenough}, the waterfilling allocation in \propref{prop:causalwater} results in all active transmitters, and  
\eqref{eq:dkwater} reduces to
\begin{align}
\frac 1 {\sigma^2_{\sX\|\sY^k}} - \frac 1 {d_k} =  \frac 1 \lambda, \label{eq:wfallactive}
\end{align}
while \eqref{eq:lambda} reduces to
\begin{align}
\lambda = K  \left(\frac 1 {\sigma_{\sX\|\sY^{[K]}}^2} - \frac 1 d \right)^{-1}. \label{eq:lambdaallactive}
\end{align}
Substituting \eqref{eq:wfallactive} into \eqref{eq:minwater} we conclude that under assumption \eqref{eq:dsmallenough},
\begin{align}
&~R_{\mathrm{CEO}}(d) \notag\\
&= \frac 1 2 \log \frac{\bar d}{d} + \frac 1 2 \sum_{k = 1}^K  \log  \left[ \left( \frac 1 {\sigma_{\sX\|\sY^k}^2} - \frac 1 {\bar d_k} \right) \lambda  \right]\label{eq:watera}\\
 &\leq \frac 1 2 \log \frac{\bar d}{d} + \frac K 2 \log  \left[ \sum_{k = 1}^K \left(   \frac 1 {\sigma_{\sX\|\sY^k}^2} - \frac 1 {\bar d_k}\right) \frac{\lambda}{K} \right]\label{eq:jensen}\\
 &= \frac 1 2 \log \frac{\bar d}{d}  + \frac K 2  \log \left( \frac 1 { \sigma_{\sX\|\sY^{[K]}}^2} - \frac 1 {\bar d} \right)\frac{\lambda}{K} \label{eq:unify}\\
 &= \frac 1 2 \log \frac{\bar d - \sigma_{\sX\| \sY^{[K]}}^2}{d - \sigma_{\sX\|\sY^{[K]}}^2} + \frac{K-1}{2} \log \frac{\bar d - \sigma_{\sX\| \sY^{[K]}}^2}{d - \sigma_{\sX\|\sY^{[K]}}^2}\frac{d}{\bar d}  
 \label{eq:activeub}
\end{align}
where
\begin{itemize}
 \item \eqref{eq:jensen} is by Jensen's inequality, since $\log$ is concave;
 \item \eqref{eq:unify} is due to 
 \begin{align}
 \frac 1 {\sigma_{\sX\|\sY^{[K]}}^2}   &= \sum_{k = 1}^K \frac 1 {\sigma_{\sX\|\sY^k}^2}  - \frac{K - 1}{\sigma_{\sX}^2}, \label{eq:sigmaxy}\\
  \frac 1 {\bar d }   &= \sum_{k = 1}^K \frac 1 {\bar d_k} - \frac{K - 1}{\sigma_{\sX}^2}, \label{eq:sigmaxdy}
\end{align}
 which holds by \lemref{lem:combo} even if the source is nonstationary  (that is, $|a| \geq 1$ and $\sigma_\sX^2 = \infty$), as a simple limiting argument taking $\frac{K - 1}{\sigma_{\sX}^2}$ to 0 confirms.
 \item \eqref{eq:activeub} holds by substituting \eqref{eq:wfallactive} into \eqref{eq:unify}.
\end{itemize}
Notice that \eqref{eq:lossc} is just another way to write  \eqref{eq:activeub}, using \eqref{eq:noisy} and \eqref{eq:noiseless}.
 To verify the condition for equality, note that `$=$' holds in \eqref{eq:watera} in the symmetrical case by \propref{prop:causalwater}, and that `$=$' holds in \eqref{eq:jensen} only in the symmetrical case due to strict concavity of the $\log$ function. 
\hspace*{\fill}$\qed$

\section{Conclusion}

In this paper, we set up the causal CEO problem (\defnref{defn:ceocode}, \defnref{def:Rceot}) and we prove that the rate-distortion function is upper bounded by the directed mutual information from the encoders to the decoder minimized subject to the distortion constraint and the separate encoding constraint, and lower bounded by the minimal directed mutual information subject to a weaker constraint (\thmref{thm:cg}). The proof of the direct coding theorem hinges upon an SLC-based nonasymptotic bound (\thmref{thm:bt}) that extends \cite[Th.~6]{yassaee2013techniqueArxiv} to the case with $K > 2$ observers and $t > 1$ time steps. An asymptotic analysis of \thmref{thm:bt} leads to an extension of the Berger-Tung inner bound  \cite{berger1978multi,tung1978multiterminal} to $t > 1$ time steps (\thmref{thm:btas}).  

By showing that the achievability bound in \thmref{thm:cg} is tight in the Gaussian case and by solving the correspoding minimal directed mutual information problem, we characterize the causal Gaussian CEO rate-distortion function as a convex optimization problem over $K$ parameters (\thmref{thm:causalceo}). We give an explicit formula in the identical-channels case (\corref{cor:sym}), and we study its asymptotic behavior as $K \to \infty$ (\corref{cor:largeK}). We derive the causal Gaussian remote rate-distortion function as a corollary to  \thmref{thm:causalceo} with $K = 1$ (\corref{cor:noisy}). Using a suboptimal waterfilling allocation over the $K$ optimization parameters in \thmref{thm:causalceo} (\propref{prop:causalwater}), we upper-bound the rate loss due to separated observers (\thmref{thm:loss}).  

We chose not to treat correlation between $n$ components of $X_i$ and $W_{i}^k$ in this paper merely to keep things simple. We expect our results to generalize to the scenario in which the  components of the source and the noise are not i.i.d. A further interesting generalization would be to consider the general vector state-space model
\begin{align}
X_{i+1} &=  A X_i + V_i \\
Y_{i}^k &=  C X_i + W_i^k,
\end{align}
where $A$ is an $n \times n$ matrix and $C$ is an $m \times n$ matrix. It will also be interesting to determine the full rate-distortion region of the causal Gaussian CEO problem as opposed to the sum rate we found in this paper. While \thmref{thm:btas} already gives an inner bound to that region, developing a converse remains open. The techniques in \cite{wagner2008improved,ekrem2014outer,wang2014vector} appear promising in that pursuit. Certain causal multiterminal source coding problems also appear within reach in view of the result in \cite{wagner2008rate} and the applicability of \thmref{thm:btas} to multiterminal source coding.

\begin{appendices}

\section{Proof of \thmref{thm:bt}}
\label{apx:btnonas}
\emph{Codebooks: }
Encoder $k$ maintains separate codebooks $\underline U_1^k, \underline U_2^k, \ldots,  \underline U_t^k$ to use at the transmission instances $1, 2, \ldots, t$ respectively. Codebook ${\underline  U}_i^k$ is an $n \times L_1^k \times \ldots \times L_{i}^k$-dimensional array: there is a separate codebook for each possible realization of past chosen codewords. 

For vector of indices $\ell_{[i]} \in \prod_{j = 1}^{i} [L_j^k]$, we denote by ${\underline  U}_i^k(\ell_{[i]})$ the codeword corresponding to index $\ell_i$, given the past indices $\ell_{[i-1]}$. For subsets $\mathcal K \subseteq [K]$ and $\mathcal I \subseteq [t]$, we denote the collection of codebooks $\underline U_{\mathcal I}^{\mathcal K} \triangleq ( \underline U_i^k \colon k \in \mathcal K, i \in \mathcal I)$. For indices $\ell_{i}^k \in [L_i^k]$, $i \in [t]$, $k \in [K]$, we denote their collection $\ell_{\mathcal I}^{\mathcal K} \triangleq (\ell_{i}^k \colon k \in \mathcal K, i \in \mathcal I)$. Finally, $\underline U_{\mathcal I}^{\mathcal K}(\ell_{\mathcal I}^{\mathcal K}) \triangleq ( \underline U_i^k(\ell_{[i]}^k) \colon k \in \mathcal K, i \in \mathcal I)$\ denotes the codewords corresponding to $\ell_{\mathcal I}^{\mathcal K}$; $1_{\mathcal I}^{\mathcal K}$ denotes the array of 1's of dimension $|\mathcal K| \times |\mathcal I |$. 

Codebook 1 for encoder $k$, $\underline U_1^k$, consists of $L_1^k$ codewords drawn i.i.d. from $P_{U_1^k}$. 
For $i = 2, \ldots, t$, codebook $i$ for user $k$, ${\underline  U}_i^k$, consists of $L_i^k$ codewords drawn i.i.d. from $P_{U_i^k | U_{[i-1]}^k = \underline U_{[i-1]}^k\left(\ell_{[i-1]}^k\right)}$, for each $\ell_{[i-1]}^k \in \prod_{j = 1}^{i-1}[L_j^k] $.

 \emph{Random binning:} Let $\mathsf B_i^k \colon [L_i^k] \mapsto [M_i^k]$, $i = 1, 2, \ldots, t$, be random mappings in which each element of $[L_i^k]$ is mapped equiprobably and independently to the set $[M_i^k]$.  
 
  We will use the notation  $\mathsf B_{\mathcal I}^{\mathcal K}(\ell_{\mathcal I}^{\mathcal K}) \triangleq ( \mathsf B_i^k(\ell_{[i]}^k) \colon k \in \mathcal K, i \in \mathcal I)$\ denotes the codewords corresponding to $\ell_{\mathcal I}^{\mathcal K}$.
   
 In the description of coding operations that follows, we denote the instances of the random codebooks in operation by $\underline u_i^k$ and those of the random binning functions by $\mathsf b_i^k$.


\emph{Encoders:} The encoders use the stochastic likelihood coder (SLC)  \cite{yassaee2013technique,yassaee2013techniqueArxiv} followed by random binning. Each user $k$ maintains a collection of encoders indexed by time $i = 1, 2, \ldots, t$; at time $i$, encoder $i$ is invoked to form and transmit a codeword at that time. 

\emph{Encoder $i$ for user $k$:} Given an observation $y_i^k \in \mathcal Y_i^k$ and past codewords $\ell_{[i-1]}^{k} \in \prod_{j = 1}^{i-1} [L_j^k]$, the SLC chooses an index $\ell_i^k \in [L_i^k]$ with probability 
\begin{align}
&~Q_{U_i^k| Y_{[i]}^k = y_{[i]}^k, U_{[i-1]}^k = \underline u_{[i-1]}^k \left(\ell_{[i-1]}^k \right)} \left(\underline u_i^k \left(\ell_{[i]}^k \right) \right)\label{eq:btenc} \\
&~= \frac{\exp\left(\imath \left(y_{[i]}^k; \underline {u}_i^k\left(\ell_{[i]}^k\right) | \underline u_{[i-1]}^k\left(\ell_{[i-1]}^k\right) \right)\right)}{\sum_{\ell = 1}^{L_i^k} \exp \left( \imath \left(y_{[i]}^k;  \underline {u}_i^k \left( \left( \ell_{[i-1]}^k, \ell\right) \right) | \underline u_{[i-1]}^k \left(\ell_{[i-1]}^k \right) \right) \right) } , \notag
\end{align}
where the conditional information density is with respect to the given distribution $P_{Y_i^k U_i^k | U_{[i-1]}^k}$. Encoder $i$ transmits $m_i^k = \mathsf b_i^k(\ell_i^k)$ to the decoder, a realization of the random variable we denote by $B_i^k$.

The causal encoder $k$ is the resulting causal probability kernel  
\begin{align}
 &~Q_{B_{[t]}^k \| Y_{[t]}^k} (m_{[t]}^k \| y_{[t]}^k) \notag\\
 &~= \sum_{\ell_{[t]}^k} \1{\mathsf b_{[t]}^k \left( \ell_{[t]}^k \right) = m_{[t]}^k}  Q_{U_{[t]}^k \| Y_{[t]}^k} (\ell_{[t]}^k \| y_{[t]}^k).
\end{align}
Since the encoders operate independently, 
\begin{align}
 Q_{U_{[t]}^{[K]} \| Y_{[t]}^{[K]}} &= \prod_{k = 1}^K Q_{U_{[t]}^k \| Y_{[t]}^k}, \\
 Q_{B_{[t]}^{[K]} \| Y_{[t]}^{[K]}} &= \prod_{k = 1}^K Q_{B_{[t]}^k \| Y_{[t]}^k}.
\end{align}

 \emph{Decoder:}
 Having received the collection of bin numbers $m_i^{[K]} \in \prod_{k =1}^K [M_i^k]$ at time $i$ and remembering  the past, the decoder invokes a generalized likelihood decoder (GLD) \cite[eq. (4)]{merhav2017gld} to select among indices that fall into those bins a collection of indices $\hat \ell_i^{[K]} \in \prod_{k = 1}^K [L_i^k]$ with probability  
\begin{align}
&~Q_{\hat U_i^{[K]} | B_{[i]}^{[K]} = b_{[i]}^{[K]}, \hat U_{[i-1]}^{[K]} = \underline u_{[i-1]}^{[K]} \left(\hat \ell_{[i-1]}^{[K]} \right) } \left(\underline u_i^{[K]} \left(\hat \ell_i^{[K]} \right) \right) = \notag\\
 &\!\! \frac{\mathsf g \left( \underline u_{[i]}^{[K]} \left( \hat \ell_{[i]}^{K} \right) \right) \1{\mathsf b_i^{[K]} \left(\hat \ell_i^{[K]} \right) = m_i^{[K]}}}
 {\sum_{\ell_i^{[K]} }  \mathsf g \left( \underline u_{[i]}^{[K]} \left( \hat \ell_{[i-1]}^{[K]}, \ell_i^{[K]}\right) \right)  \1{\! \mathsf b_i^{[K]}(\ell_i^{[K]}) = m_i^{[K]} \! } }, \label{eq:btdec}
\end{align}
where 
\begin{align}
 \mathsf g \left(u_{[i]}^{[K]} \right) \triangleq \prod_{k = 1}^K \1{\!\jmath^{\pi(k)} \left(u_{[i]}^{\pi([K])} \right) \! \geq \! \log \frac{L_i^{\pi(k)} }{M_i^{\pi(k)}}  + \beta_i^{\pi(k)} \! }. \label{eq:gld}
\end{align}

Having determined $\hat \ell_i^{[K]}$, the decoder applies the given transformation $P_{\hat X_i | U_{[i]}^{[K]}, \hat X_{[i-1]}}$ to form the estimate of the source 
$\hat X_i \left(\underline u_{[i]}^{[K]} \left(\hat \ell_{[i]}^{[K]}\right) \right)$.
The causal decoder is the resulting causal kernel $Q_{\hat X_{[t]} \| B_{[t]}^{[K]}}$.

 \emph{Error analysis:}
  We consider two error events: 
\begin{align}
&\mathcal E_{\mathrm{dec}} \colon \hat U_{[t]}^{[K]} \neq U_{[t]}^{[K]}    
\\
&\mathcal E_{\mathrm{enc}} \colon \bigcup_{i = 1}^t \left\{ \sd \left(X_i, \hat X_i \left(U_{[i]}^{[K]} \right) \right) > d_i   \right\}, 
\end{align}
where $U_{[i]}^{[K]}$ are the codewords chosen by the encoders at encoding step \eqref{eq:btenc}, and $\hat U_{[t]}^{[K]} $ is the decoder's estimate of those codewords after decoding step \eqref{eq:btdec}. Note that $\mathcal E_{\mathrm{dec}}$ is the event that some codewords are not recovered (decoding error), and $\mathcal E_{\mathrm{enc}}$ is the event that some distortions exceed threshold even if all the codewords are recovered correctly (encoding error). We denote for brevity by $\mathcal F$ the sigma-algebra generated by $Y_{[t]}^{[K]}$, $\underline U_{[t]}^{[K]} \left(1_{[t]}^{[K]}\right), \mathsf B_{[t]}^{[K]}\left(1_{[t]}^{[K]}\right)$, $\hat X_{[i]}\left( \underline U_{[i]}^{[K]} \left(  1_{[i]}^{[K]} \right) \right)$; by $\mathbb Q$ the probability measure generated by the code; and by $F_i^k$, $G_i$  the denominators in \eqref{eq:btenc} and \eqref{eq:btdec}, respectively. Following Shannon's random coding argument and the Jensen inequality technique of Yassaee et al. \cite{yassaee2013technique,yassaee2013techniqueArxiv}, we proceed to bound an expectation of the indicator of the correct decoding event with respect to both the actual source code and the random codebooks.  
  \begin{align}
&~ \E{\mathbb Q\left[ \prod_{i = 1}^t \1{ \sd \left(X_i, \hat X_i \right) \leq d_i  \mid \underline U_{[t]}^{[K]}, \mathsf B_{[t]}^{[K]} } \right] } \notag\\
\geq&~ \E{\mathbb Q\bigg[ \mathcal E_{\mathrm{enc}}^c \cap \mathcal E_{\mathrm{dec}}^c \mid \underline U_{[t]}^{[K]}, \mathsf B_{[t]}^{[K]} \bigg] } \\
=
&~ \mathbb E \bigg[ \sum_{\ell_{[t]}^{[K]} \in \prod_{i = 1}^t \prod_{k = 1}^K [L_i^k]} Q_{U_{[t]} \| Y_{[t]}^{[K]}}\left(\underline U_{[t]}^{[K]} \left(\ell_{[t]}^{[K]}\right)  \| Y_{[t]}^{[K]} \right)  \notag \\
& \cdot \sum_{m_{[t]}^{[K]} \in \prod_{i = 1}^t \prod_{k = 1}^K [M_i^k] }\1{\mathsf B_{[t]}^{[K]} \left( \ell_{[t]}^{[K]} \right) = m_{[t]}^{[K]}}    \notag \\
&\cdot Q_{\hat U_{[t]}^{[K]} \| B_{[t]}^{[K]} = m_{[t]}^{[K]}} \left(\underline U_{[t]}^{[K]}\left(\ell_{[t]}^{[K]}\right)\right)
 1\{ \mathcal E_{\mathrm{enc}}^c\}  \bigg]  \label{eq:expub} \\
=&~ \prod_{k = 1}^K \prod_{i = 1}^t M_i^k L_i^k  \notag \\
 &\cdot \mathbb E \bigg[ \mathbb E \bigg[ Q_{U_{[t]}^{[K]} \| Y_{[t]}^{[K]}}\left(\underline U_{[t]}^{[K]} \left(1_{[t]}^{[K]}\right)  \| Y_{[t]}^{[K]} \right)  \notag \\
 &\cdot  \1{\mathsf B_{[t]}^{[K]} \left(1_{[t]}^{[K]}\right) = 1_{[t]}^{[K]} }      \notag\\
 &\cdot Q_{\hat U_{[t]}^{[K]} \| B_{[t]}^{[K]} = 1_{[t]}^{[K]}} \left(\underline U_{[t]}^{[K]} \left(1_{[t]}^{[K]}\right) \right)  
 1\{ \mathcal E_{\mathrm{enc}}^c\}   
 \mid \mathcal F   \bigg] \bigg]\!\! \label{eq:randsym}  \\
\geq&~ \prod_{k = 1}^K \prod_{i = 1}^t M_i^k L_i^k    \notag \\
&\cdot \mathbb E \Bigg[\prod_{k = 1}^K  \prod_{i = 1}^t \frac{  \exp\left(\imath \left(Y_{[i]}^k; \underline {U}_i^k\left(1_{[i]}\right) \mid \underline U_{[i-1]}^k\left(1_{[i-1]}\right) \right)\right)}{\E{F_i^k \mid \mathcal F}} \notag\\
 &\cdot \frac{\mathsf g \left( \underline U_{[i]}^{[K]} \left(1_{[i]}^{[K]} \right) \right)  \1{\mathsf B_i^{[K]} \left(1^{[K]} \right) = 1^{[K]}} }{  \E{G_i \mid \mathcal F} }  \notag\\
&\cdot \1{ \sd \left(X_i, \hat X_i \left( \underline U_{[i]}^{[K]} \left(  1_{[i]}^{[K]} \right) \right) \right) \leq d_i}  \Bigg], \label{eq:mainbt}
\end{align}
where 
\begin{itemize}
 \item the expectation $\mathbb E$ in \eqref{eq:expub} is with respect to the codebooks $\underline U_{[t]}^{[K]}$, the random binning functions $\mathsf B_{[t]}^{[K]}$, the decoder $P_{\hat X_{[t]} \| U_{[t]}^{[K]}}$ and $X_{[t]}$, $Y_{[t]}^{[K]}$; 
 \item \eqref{eq:randsym} uses that both the codewords and the binning functions for the $i$-th time instant are independently and identically distributed, thus each choice of $\ell_{[t]}^{[K]}$ and $m_{[t]}^{[K]}$ results in the same probability as the choice $\ell_{[t]}^{[K]} = 1_{[t]}^{[K]}$ and $m_{[t]}^{[K]} = 1_{[t]}^{[K]}$. Here we also conditioned on $\mathcal F$ before taking an outer expectation with respect to it, which will facilitate the next step of the calculation.
 \item the main step \eqref{eq:mainbt} is shown as follows. The product $Q_{U_{[t]}^{[K]} \| Y_{[t]}^{[K]}} Q_{\hat U_{[t]}^{[K]} \| B_{[t]}^{[K]}}$ is proportional to the product of $(K + 1)t$ factors $ \prod_{i = 1}^t \frac 1 {G_i} \prod_{k = 1}^K \frac 1 {F_i^k}$. Applying Jensen's inequality to this jointly convex function of $(K + 1)t$ variables yields
\end{itemize}
\begin{align}
\E{ \prod_{i = 1}^t \frac 1 {G_i} \prod_{k = 1}^K \frac 1 {F_i^k}  \mid \mathcal F  } \geq \prod_{i = 1}^t \frac 1 {\E{ G_i \mid \mathcal F } } \prod_{k = 1}^K \frac 1 {\E{ F_i^k \mid \mathcal F }}
\label{eq:yassjen}
\end{align}

We compute each factor in  \eqref{eq:yassjen} as follows. 
\begin{align}
&~\E{F_i^k | \mathcal F} \\
=&~ \!\! \E{\sum_{\ell = 1}^{L_i^k} \exp \left( \imath \left(Y_{[i]}^k;  \underline {U}_i^k \left(1_{[i-1]}, \ell \right) | \underline U_{[i-1]}^k \left(1_{[i-1]} \right) \right) \right) \!\! \mid \!\! \mathcal F } \notag\\
=&~ \exp \left( \imath \left(Y_{[i]}^k;  \underline {U}_i^k \left(1_{[i]} \right) | \underline U_{[i-1]}^k \left(1_{[i-1]} \right) \right) \right) + (L_i^k - 1) \cdot \notag\\
&~  \E{\exp \left( \imath \left(Y_{[i]}^k;  \underline {U}_i^k(1_{[i-1]}, 2) | \underline U_{[i-1]}^k \left(1_{[i-1]} \right) \right) \right) \mid \mathcal F} \label{eq:Fia} \\
=&~ \exp \left( \imath \left(Y_{[i]}^k;  \underline {U}_i^k(1_{[i]}) | \underline U_{[i-1]}^k \left(1_{[i-1]} \right) \right) \right) + (L_i^k - 1), \label{eq:Fi}
\end{align}
where to write \eqref{eq:Fia} we used that the codewords $\{\underline {U}_i^k(1_{[i-1]}, \ell)\colon \ell \neq 1\}$ are identically distributed conditioned on $\mathcal F$.

To evaluate $\E{G_i | \mathcal F}$, we partition the set of all $\ell_i^{[K]} \in \prod_{i= 1}^K [L_i^k]$  into index sets parameterized by $\mathcal K \subseteq [K]$:
\begin{align}
\mathcal L_i(\mathcal K) \triangleq \Big\{ \ell^{[K]} \in \prod_{i= 1}^K [L_i^k] \colon  & \ell^{\pi(k)} = 1, k \in \mathcal K, \notag\\
& \ell^{\pi(k)} \neq 1, k \in \mathcal K^c \Big\}, 
\end{align}
and for each $\ell_i^{[K]} \in \mathcal L_i(\mathcal K)$, $\mathcal K \subset K$, we upper-bound $\mathsf g( \cdot)$ as  
\begin{align}
&~\mathsf g \left(\underline u_{[i]}^{[K]} \left( 1_{[i-1]}^{[K]}, \ell_i^{[K]} \right) \right) \\
\leq & \prod_{k \in \mathcal K^c }
\1{\jmath^{\pi(k)} \left(\underline u_{[i]}^{\pi([K])} \right)  \geq \log \frac{L_i^{\pi(k)} }{M_i^{\pi(k)}}  + \beta_i^{\pi(k)} }\\
\leq& \prod_{k \in \mathcal K^c} \!\!\frac{M_i^{\pi(k)}}{L_i^{\pi(k)}} \exp\left( \jmath^{\pi(k)} \left( \underline u_{[i]}^{\pi([K])} \left( 1_{[i-1]}^{[K]}, \ell_i^{[K]} \right) \right) - \beta_{i}^{\pi(k)} \! \right), \label{eq:gub}
\end{align}
while for $\mathcal K = [K]$, we upper-bound it as
\begin{align}
 \mathsf g \left(\underline u_{[i]}^{[K]} \left( 1_{[i]}^{[K]} \right) \right) &\leq 1. \label{eq:gub0}
\end{align}

Note that for each $\ell_i^{[K]} \in \mathcal L_i(\mathcal K)$, $\mathcal K \subset K$
\begin{align}
 \E{\prod_{k \in \mathcal K^c}  \exp\left( \jmath^{\pi(k)} \left(\underline U_{[i]}^{\pi([K])} \left( 1_{[i-1]}^{[K]}, \ell_i^{[K]} \right) \right) \right) | \mathcal F} = 1. \label{eq:gubexp}
\end{align}
The upper-bound in \eqref{eq:gub} and the equality in \eqref{eq:gubexp} are key to the analysis of our GLD \eqref{eq:btdec}. 

Now, $\E{G_i | \mathcal F}$ is bounded as 
\begin{align}
&~ \E{G_i | \mathcal F} = \notag\\
&~ \mathbb E \Bigg[ \sum_{\ell_i^{[K]} \in \prod_{k = 1}^K [L_i^k] }  \mathsf g \left(\underline U_{[i]}^{[K]} \left( 1_{[i-1]}^{[K]}, \ell_i^{[K]} \right) \right)   \notag\\
  &~ \hspace{70pt}\cdot \1{\mathsf B_i^{[K]}(\ell_i^{[K]}) = 1^{[K]}}\mid \mathcal F \Bigg] \\
=&~ \mathsf g \left(\underline U_{[i]}^{[K]} \left( 1_{[i]}^{[K]} \right) \right) \1{\mathsf B_i^{[K]}(1^{[K]}) = 1^{[K]}} 
\notag\\
&+ \sum_{\substack{\mathcal K \subset [K]}}  \E{ \sum_{ \ell_i^{[K]} \in \mathcal L_i(\mathcal K) }\mathsf g \left(\underline U_{[i]}^{[K]} \left( 1_{[i-1]}^{[K]}, \ell_i^{[K]} \right) \right) \mid \mathcal F }\notag\\
&\cdot  \prod_{k \in \mathcal K^c} \frac{1}{M_i^{\pi(k)}} \cdot \1{\mathsf B_i^{ \pi(\mathcal K)}(1^{\mathcal K}) = 1^{\mathcal K}} \\
\leq&~ \1{\mathsf B_i^{[K]}(1^{[K]}) = 1^{[K]}}  \label{eq:Gi}\\
&+ \sum_{\substack{\mathcal K \subset [K]}} \exp\left( - \sum_{k \in \mathcal K^c} \beta_{i}^{\pi(k)} \right) \1{\mathsf B_i^{ \pi(\mathcal K)}(1^{ \mathcal K}) = 1^{\mathcal K}}, 
\notag
\end{align} 
where \eqref{eq:Gi} follows from \eqref{eq:gub}, \eqref{eq:gub0} and \eqref{eq:gubexp}. 

Now, plugging \eqref{eq:Fi} and \eqref{eq:Gi} into \eqref{eq:mainbt} and computing the expectation in \eqref{eq:mainbt} with respect to the codebooks and the binning functions,  we conclude that the probability of successful decoding is bounded below as
\begin{align}
&~ 1 - \epsilon \geq    \label{eq:bt}\\
&~\mathbb E \Bigg[\prod_{k = 1}^K  \prod_{i = 1}^t \frac{  1}{ \frac 1 {L_i^{k}} \exp \left( \imath \left(Y_{[i]}^k;  U_i^k | U_{[i-1]}^k \right) \right) + \left(1 - \frac 1 {L_i^{k}}\right)}  \notag\\
 &\cdot \frac{\mathsf g \left(U_{[i]}^{[K]} \right) \1{ \sd \left(X_i, \hat X_i\left(U_{[i]}^{[K]} \right) \right) \leq d_i} }{ 1 + \sum_{\substack{\mathcal K \subset [K]}} \exp\left( - \sum_{k \in \mathcal K^c} \beta_{i}^k \right) } \Bigg] \notag.
\end{align}

\emph{ Loosening the bound \eqref{eq:bt}:}
Here we again follow the recipe of Yassaee et al. \cite{yassaee2013technique,yassaee2013techniqueArxiv}.

\begin{align}
 &~ 1 - \epsilon \geq    \notag\\
&~\mathbb E \Bigg[\prod_{k = 1}^K  \prod_{i = 1}^t \frac{  1}{  (L_i^{k})^{-1} \exp \left( \imath \left(Y_{[i]}^k;  U_i^k | U_{[i-1]}^k \right) \right) + 1}  \notag\\
 &\cdot \frac{\mathsf g \left(U_{[i]}^{[K]} \right) \1{ \sd \left(X_i, \hat X_i\left(U_{[i]}^{[K]} \right) \right) \leq d_i} }{  \sum_{\substack{\mathcal K \subseteq [K]}} \exp\left( - \sum_{k \in \mathcal K} \beta_{i}^k \right) } \Bigg] \label{eq:bta}\\
    &\geq  \prod_{k = 1}^K  \prod_{i = 1}^t  \frac{  \Prob{\mathcal E^c}}{  \left[ 1 + \exp(-\alpha_i^k) \right] \left[ \sum_{\substack{\mathcal K \subseteq [K]}} \exp\left( - \sum_{k \in \mathcal K} \beta_{i}^k \right) \right]   }    \label{eq:wel} 
     \end{align}
 where \eqref{eq:bta} holds by weakening \eqref{eq:bt} using $1 - \left(L_t^k\right)^{ -1} \leq 1$ and rewriting for brevity 
\begin{align}
\!\! 1 + \sum_{\mathcal K \subset [K]} \exp\left( - \sum_{k \in \mathcal K^c} \beta_{i}^k \right)&=\sum_{\mathcal K \subseteq [K]} \exp\left( - \sum_{k \in \mathcal K} \beta_{i}^k \right);
\end{align}
 \eqref{eq:wel} is obtained by weakening \eqref{eq:bta} by multiplying the random variable inside the expectation by  $\1{\mathcal E^c}$ and using the conditions in $\mathcal E$ \eqref{eq:event} to upper-bound $\imath \left(Y_{[i]}^k;  U_i^k | U_{[i-1]}^k \right)$ in the denominator. 
 
Rewriting \eqref{eq:wel}, we obtain 
\begin{align}
 \epsilon &\leq 1 - \\
 &~\prod_{k = 1}^K  \prod_{i = 1}^t  \frac{  \Prob{\mathcal E^c}}{  \left[ 1 + \exp(-\alpha_i^k) \right] \left[ \sum_{\mathcal K \subseteq K } \exp(-\sum_{k \in \mathcal K} \beta_i^{k}) \right]   }   \notag  \\
 &=  
 \Prob{\mathcal E} + \gamma  \, \Prob{ \mathcal E^c} \\
 &\leq \Prob{\mathcal E} + \gamma.
\end{align}

 \section{Proof of \thmref{thm:btas}}
 \label{sec:btas}
 We analyze the bound in \thmref{thm:bt} with
\begin{align}
P_{U_{[t]}^k \| Y_{[t]}^k} &= P_{\mathsf U_{[t]}^k \| \mathsf Y_{[t]}^k}^{\otimes n}, \\
P_{ \hat X_{[t]}^{[K]} \| U_{[t]}^{[K]}} &=  P_{ \hat {\mathsf X}_{[t]}^{[K]} \| \mathsf U_{[t]}^{[K]}}^{\otimes n},
\end{align}
single-letter kernels chosen so that
\begin{align}
\E{ \sd \left(\mathsf X_i, \hat {\mathsf X}_i \left( \mathsf U_{[i]}^{[K]}\right)  \right)} = d_i + \delta,
 \end{align}
 for some $\delta > 0$. We also fix an arbitrary permutation $\pi \colon [K] \mapsto [K]$. 
 Denote for brevity the divergences 
\begin{align}
&~D_i^{\pi(k)} \triangleq  \E{\jmath^{\pi(k)} \left(\mathsf U_{[i]}^{\pi([K])} \right)} \\
=&~ D \left(\!  P_{\mathsf U_{i}^{\pi(k)} | \mathsf U_{i}^{\pi([k-1])} \mathsf U_{[i-1]}^{\pi([K])} }  \| P_{\mathsf U_{i}^{\pi(k)} | \mathsf U_{[i-1]}^{\pi(k)} }  | P_{\mathsf U_{i}^{\pi([k-1])} \mathsf U_{[i-1]}^{\pi([K])}} \! \right)\notag
\end{align}
For $k \in [K]$, $i \in [t]$, let 
\begin{equation}
 \alpha_i^k = \beta_i^k = n \delta,
\end{equation}
and choose $L_i^k$, $M_i^k$ to satisfy
\begin{align}
\log L_i^k &\geq n\, I \left(\mathsf Y_{[i]}^k; \mathsf U_i^k |\mathsf U_{[i-1]}^k \right) + 2 \alpha_i^k, \label{eq:Lik}\\
 \log M_i^{\pi(k)} &\geq  \log L_i^{\pi(k)} - n\, D_i^{\pi(k)} + 2 \beta_i^k. \label{eq:Mik}
\end{align}
Note that since $\mathsf U_i^k - \left( \mathsf Y_{[i]}^k, \mathsf U_{[i-1]}^k \right) - \mathsf U_{[i]}^{[K] \backslash \{k\}}$, it holds that
\begin{align}
& I \left(\mathsf Y_{[i]}^{\pi(k)}; \mathsf U_{i}^{\pi(k)} | \mathsf U_{i}^{\pi([k-1])}, \mathsf U_{[i-1]}^{\pi([K])} \right) \notag\\
=&  I \left(\mathsf Y_{[i]}^{\pi(k)}; \mathsf U_i^{\pi(k)} |\mathsf U_{[i-1]}^{\pi(k)} \right)  - D_i^{\pi(k)} \label{eq:diffDk},
\end{align}
and thus summing both sides of \eqref{eq:Mik} over $i \in [t]$ we obtain (cf. \eqref{eq:Rbt})
\begin{align}
\frac 1 n \sum_{i = 1}^t \log M_i^k  \geq  &~ I \left({\mat Y}_{[t]}^{\pi(k)} \to  \mat U_{[t]}^{\pi(k)}  \| \mat U_{[t]}^{\pi([k-1])}, \D \mat  U_{[t]}^{[K]}  \right) \notag \\
&+ 4 t \delta.
\end{align}
Applying the union bound to $\Prob{\mathcal E}$ and the law of large numbers to each of the resultant $(2 K + 1)t$ terms, we further conclude that $\Prob{\mathcal E} \to 0$ as $n \to \infty$. Furthermore, $\gamma \to 0$ as $n \to \infty$, and therefore by \thmref{thm:bt} there exists a sequence of codes with $\log L_i^k$ and $\log M_i^k$ satisfying \eqref{eq:Lik}, \eqref{eq:Mik} with excess-distortion probability $\epsilon \to 0$ as $n \to \infty$. 

Under our assumption on the $p$-th moment of the distortion measure \eqref{eq:dmom}, the existence of an $(M_{[t]}^{[K]}, d_{[t]}, \epsilon)$ excess-distortion code with $\frac 1 t \sum_{i = 1}^t d_i \leq d$
implies the existence of an $(M_{[t]}^{[K]}, d(1 - \epsilon) + d_p \epsilon^{1 - 1/p})$ average distortion code via a standard argument using H\"older's inequality \cite[Th. 25.5]{polyanskiy2012notes}.

 \section{Two characterizations of Berger-Tung bound}
\label{apx:bt}
\begin{prop}
The region $\mathcal R$ in \eqref{eq:btreg} 
 is equivalent to the region $\mathcal R^\prime$ in \eqref{eq:btperm}.
\label{prop:btequiv}
\end{prop}
\begin{proof}[Proof of \propref{prop:btequiv}]
Observe that any subset $\mathcal A$ of $[K]$ with cardinality $k$ is equal to $\pi([k])$, for some permutation $\pi$ on $[K]$. 

First, we show that $\mathcal R^\prime \subseteq \mathcal R$. Fix $\pi$ and consider $\mathcal K = \pi([k])$.
Since given $\sY_k$, $\sU_k$ is independent of $\sU^{[K] \backslash \{k\}}$,
\begin{align}
I(\sY^{\mathcal K}; \sU^{\mathcal K}) &= \sum_{j = 1}^k I(\sY_{\pi(j)}; \sU_{\pi(j)} | \sU^{\pi[j-1]}),  \\
I(\sY^{\mathcal K^c}; \sU^{\mathcal K^c} | \sU^{\mathcal K}) &= \sum_{j = k+1}^K I(\sY_{\pi(j)}; \sU_{\pi(j)} | \sU^{\pi[j-1]}).  \label{eq:btperm2}
\end{align}
From \eqref{eq:btperm2}, we conclude that any set of rates that satisfies \eqref{eq:btperm} for $\pi$ must also satisfy \eqref{eq:btreg} for $\mathcal A = \mathcal K^c$. Thus, $\mathcal R^\prime \subseteq \mathcal R$.

To show that $\mathcal R \subseteq \mathcal R^\prime$, note, using  the operational Markov chain condition $\sU^{\mathcal B} - \sY^{\mathcal B} - \sY^{\mathcal A \backslash \mathcal B} - \sU^{\mathcal A \backslash \mathcal B}$, that for all 
 $\mathcal  B \subseteq \mathcal A$,
\begin{align}
 I(\sY^{\mathcal A}; \sU^{\mathcal A} ) = I(\sY^{\mathcal A \backslash \mathcal B}; \sU^{\mathcal A \backslash \mathcal B } | \sU^{\mathcal B}) + I(\sY^{\mathcal B}; \sU^{\mathcal B} ). \label{eq:bteq}
\end{align}
Since
\begin{align}
\begin{cases}
 S_1 \geq I_1 \\
 S_1 + S_2 \geq I_1 + I_2  
\end{cases}
\Longleftrightarrow 
\begin{cases}
 S_1 \geq I_1 \\
 S_2 \geq I_2  
\end{cases},
\end{align}
\eqref{eq:bteq} implies that for any $\mathcal A \subseteq [K]$,
\begin{align}
\begin{cases}
\sum_{k \in \mathcal A^c} R^k \geq I(\sY^{\mathcal A^c}; \sU^{\mathcal A^c} | \sU^{\mathcal A})  \\
\sum_{k \in [K]} R^k \geq I(\sY^{[K]}; \sU^{[K]})
\end{cases}\\
\Longleftrightarrow 
\begin{cases}
\sum_{k \in \mathcal A} R^k \geq I(\sY^{\mathcal A}; \sU^{\mathcal A} )\\
\sum_{k \in \mathcal A^c} R^k \geq I(\sY^{\mathcal A^c}; \sU^{\mathcal A^c} | \sU^{\mathcal A}) 
\end{cases} 
\end{align}
and for any $\mathcal B \subseteq \mathcal A$, 
\begin{align}
&\begin{cases}
\sum_{k \in \mathcal B} R^k \geq I(\sY^{\mathcal B}; \sU^{\mathcal B}) \\
\sum_{k \in \mathcal A} R^k \geq I(\sY^{\mathcal A}; \sU^{\mathcal A} )
\end{cases} \\
\Longleftrightarrow
&\begin{cases}
\sum_{k \in \mathcal B} R^k \geq I(\sY^{\mathcal B}; \sU^{\mathcal B}) \\
\sum_{k \in \mathcal A \backslash \mathcal B} R^k \geq I(\sY^{\mathcal A \backslash \mathcal B}; \sU^{\mathcal A \backslash \mathcal B} | \sU^{\mathcal B}) \label{eq:bteq2}
\end{cases}
\end{align}
For $\mathcal B = \pi([k-1])$ and $\mathcal A = \pi([k])$, the second inequality in \eqref{eq:bteq2} is exactly the inequality \eqref{eq:btperm}.  Since any set of rates satisfying \eqref{eq:btreg} must also satisfy \eqref{eq:bteq2} for all $\mathcal B \subseteq \mathcal A \subseteq [K]$,  we conclude that $\mathcal R \subseteq \mathcal R^\prime$. 
\end{proof}

\section{MMSE estimation lemmas}
\label{apx:gest}
Lemmas \ref{lem:back} and \ref{lem:combo} are corollaries to the following result. 
\begin{lemma}
Let $X \sim \Gauss{0, \sigma_X^2 }$, and let
\begin{align}
Y_k = X + W_k,~ k = 1, \ldots, K, \label{eq:broadcast}
\end{align}
where $W_k \sim \Gauss{0, \sigma_{W_k}^2 }$, $W_k \perp W_j$, $j \neq k$. 
 Then, the MMSE estimate and the normalized estimation error of $X$ given $Y_{[K]}$ are given by
\begin{align}
\E{X | Y_{[K]}} &= \sum_{k = 1}^K  \frac{\sigma_{X|Y_{[K]}}^2}{\sigma_{W_k}^2} Y_k, \label{eq:estk}\\
\frac 1 {\sigma_{X|Y_{[K]}}^2 } &= \frac 1 {\sigma_X^2} + \sum_{k =1}^K \frac 1 {\sigma_{W_k}^2}. \label{eq:msek}
\end{align}
\label{lem:Gest}
\end{lemma}
\begin{proof}[Proof of \lemref{lem:Gest}]
The result is well known; we provide a proof for completeness.
For jointly Gaussian random vectors $X, Y$, 
 \begin{align}
\E{X| Y = y} &= \E{X} + \Sigma_{XY} \Sigma_{YY}^{-1}\left( y - \E{Y}\right),\\
\Cov[X|Y] &= \Sigma_{XX} - \Sigma_{XY} \Sigma_{YY}^{-1} \Sigma_{YX}.
\end{align}

Denote for brevity
\begin{align}
\Sigma_{W} \triangleq \begin{bmatrix}
\sigma_{W_1}^2 & & 0\\
& \ddots &\\
0 & &\sigma_{W_K}^2
\end{bmatrix}.  
\end{align}
In our case, $X$ is a scalar and $Y = Y_{[K]}$ is a vector, and
\begin{align}
\Sigma_{XX} &= \sigma_X^2,\\
\Sigma_{YY} 
&=\Sigma_W
+ 
\begin{bmatrix}
1\\
\vdots
\\
1
\end{bmatrix} 
\sigma_X^2
\begin{bmatrix}
1 & \cdots & 1
\end{bmatrix},\\
\Sigma_{XY} &=  
\sigma_X^2
\begin{bmatrix}
1 &\ldots & 1
\end{bmatrix}.
\end{align}
Using the matrix inversion lemma, we compute readily
\begin{align}
&~ \Cov[X|Y]^{-1} =
 \Sigma_{XX}^{-1} -  \Sigma_{XX}^{-1} \Sigma_{XY} \notag\\
 &\left( \Sigma_{YX} \Sigma_{XX}^{-1} \Sigma_{XY} - \Sigma_{YY} \right)^{-1} \Sigma_{YX} \Sigma_{XX}^{-1} \\
 =&~ \Sigma_{XX}^{-1} +  \Sigma_{XX}^{-1} \Sigma_{XY} 
\Sigma_{W}
^{-1}
\Sigma_{YX} \Sigma_{XX}^{-1} \\
=&~ \frac 1 {\sigma_X^2} +  \frac 1 {\sigma_{W_1}^2} + \ldots + \frac 1 {\sigma_{W_K}^2},
\end{align}
which shows \eqref{eq:msek}.
To show \eqref{eq:estk}, we apply the matrix inversion lemma to $\Sigma_{YY}$ to write: 
\begin{align}
 \Sigma_{YY}^{-1} = \Sigma_W^{-1} - \Sigma_W^{-1}
 \begin{bmatrix}
1\\
\vdots
\\
1
\end{bmatrix} 
\sigma_{X|Y_{[K]}}^2
\begin{bmatrix}
1 &\ldots & 1
\end{bmatrix} \Sigma_W^{-1}.
\end{align}
It's easy to verify that
\begin{align}
&~\sigma_X^2 \begin{bmatrix}
1 &\ldots & 1
\end{bmatrix} 
\left( I_n \frac 1 {\sigma_{X|Y_{[K]}}^2}   - \Sigma_W^{-1}
 \begin{bmatrix}
1\\
\vdots
\\
1
\end{bmatrix} 
\begin{bmatrix}
1 &\ldots & 1
\end{bmatrix} \right)\notag \\
=&~ \begin{bmatrix}
1 &\ldots & 1
\end{bmatrix}, 
\end{align}
where $I_n$ is the $n \times n$ identity matrix, so 
 \begin{align}
\E{X| Y = y} &=  \Sigma_{XY} \Sigma_{YY}^{-1}y\\
&= \begin{bmatrix}
1 &\ldots & 1
\end{bmatrix} 
\Sigma_W^{-1}
\sigma_{X|Y_{[K]}}^2 y,
\end{align}
which is equivalent to \eqref{eq:estk}. 
\end{proof}

\begin{proof}[Proof of \lemref{lem:back}]
 Equality \eqref{eq:cond1} follows from
 \begin{align}
\sigma_Y^2 &= \sigma_{X}^2 + \sigma_{W}^2,\\
\frac 1 {\sigma_{X | Y}^2} &= \frac 1 {\sigma_{X}^2}  + \frac 1 {\sigma_{W}^2}, \label{eq:msekpart}
\end{align}
where \eqref{eq:msekpart} is a particularization of \eqref{eq:msek}. 
\end{proof}

\begin{proof}[Proof of \lemref{lem:combo}]
Notice that \eqref{eq:backx} with $\bar X_k = \E{X | Y_k}$ and $W_k^\prime \sim \mathcal N(0, \sigma_{X|Y_k}^2)$ is just another way to write \eqref{eq:broadcast}. Reparameterizing \eqref{eq:estk} and \eqref{eq:msek} accordingly, one recovers \eqref{eq:combo} and \eqref{eq:msecombo}. 
\end{proof}

\begin{remark}
We may use \lemref{lem:Gest} to derive the Kalman filter for the estimation of $\sX_i$ \eqref{eq:xi} given the history of observations  $\sY_{[i]}^{[K]}$ \eqref{eq:yik}: 
 \begin{align}
\bar \sX_i &= a \bar \sX_{i-1} + \sum_{k = 1}^K  \frac{\sigma_{\sX_i|\sY_{[i]}^{[K]}}^2}{\sigma_{\mathsf W_k}^2} \left(\sY_i^{k} - a \bar \sX_{i-1} \right), \label{eq:estkcausal}\\
\frac 1 {\sigma_{\sX_i|\sY_{[i]}^{[K]}}^2 } &= \frac 1 {\sigma_{\sX_i|\sY_{[i-1]}^{[K]}}^2 } + \sum_{k =1}^K \frac 1 {\sigma_{\mathsf W_k}^2}. \label{eq:msekcausal}
\end{align}
 where $\bar  \sX_i$ is defined in \eqref{eq:Xibar1let}. Equation \eqref{eq:estkcausal} is the Kalman filter recursion with Kalman filter gain equal to the row vector $ \sigma_{\sX_i| \sY_{[i]}^{[K]}}^2 \left( \frac 1  {\sigma_{\mathsf W_1}^2}, \ldots, \frac 1 {\sigma_{\mathsf W_K}^2} \right)$, and \eqref{eq:msekcausal} is the corresponding Riccati recursion for the MSE\detail{ (Appendix~\ref{apx:kalman})}.
\end{remark}

\detail{
\section{Verification of \eqref{eq:estkcausal} and \eqref{eq:msekcausal}}
\label{apx:kalman}
We verify \eqref{eq:estkcausal} and \eqref{eq:msekcausal} using the Kalman filter equations. Denote the observation matrix 
\begin{align}
\mat C \triangleq 
\begin{bmatrix}
1 \\ \vdots \\ 1
\end{bmatrix},
\end{align}
 the covariance matrix of the observation noise
\begin{align}
\mat  \Sigma_W = \diag (\sigma_{W_1}^2, \ldots, \sigma_{W_K}^2), 
\end{align}
and the error (co)variances
\begin{align}
p_{i | i - 1} &=  \sigma_{X_i| Y_{[i-1]}}^2,\\
p_{i | i } &= \sigma_{X_i| Y_{[i]}}^2.
\end{align}
Innovation of the observation is $\tilde Y_i = Y_i - a \mat C  X_{i-1}$, its covariance is
\begin{align}
\mat S_i = \mat C \mat C^T  p_{i | i-1} + \mat  \Sigma_W. 
\end{align}
Matrix inversion lemma yields
\begin{align}
 \mat S_i^{-1} = \mat  \Sigma_W^{-1} - \mat \Sigma_W^{-1} \mat C \left( p_{i|i-1}^{-1} + \underbrace{\mat C^T \mat \Sigma_W^{-1} \mat C}_{\sum_{k =1}^K \frac 1 {\sigma_{W_k}^2}}\right)^{-1} \mat C^T \mat \Sigma_W^{-1} \label{dtl:Si}
\end{align}
Kalman filter gain is
\begin{align}
 \mat K_i &= p_{i | i-1} \mat C^T \mat S_i^{-1}\\
 &=p_{i | i-1} \left( 1  - \frac{\sum_{k =1}^K \frac 1 {\sigma_{W_k}^2}}{\frac 1 {p_{i | i-1}} + \sum_{k = 1}^K  \frac 1 {\sigma_{W_k}^2}  }\right) \left( \frac 1  {\sigma_{W_1}^2}, \ldots, \frac 1 {\sigma_{W_K}^2} \right)\\
 &= \frac 1 {\frac 1 {p_{i | i-1}} + \sum_{k = 1}^K  \frac 1 {\sigma_{W_k}^2}}\left( \frac 1  {\sigma_{W_1}^2}, \ldots, \frac 1 {\sigma_{W_K}^2} \right) \\
 &= p_{i|i}\left( \frac 1  {\sigma_{W_1}^2}, \ldots, \frac 1 {\sigma_{W_K}^2} \right), 
 \label{dtl:Ki}
\end{align}
where \eqref{dtl:Ki} is shown below. 
Note that
\begin{align}
 \mat K_i \mat C &= \frac{\sum_{k = 1}^K \frac 1 {\sigma_{W_k}^2} }{ \frac 1 {p_{i | i-1}} + \sum_{k = 1}^K  \frac 1 {\sigma_{W_k}^2}  }
\end{align}
and
\begin{align}
 1 - \mat K_i \mat C = \frac{\frac 1 {p_{i | i-1}} }{ \frac 1 {p_{i | i-1}} + \sum_{k = 1}^K  \frac 1 {\sigma_{W_k}^2}  }.
\end{align}

The Riccati equation for the me is given by
\begin{align}
p_{i | i} &= (1 - \mat K_i \mat C) p_{i | i-1}\\
&= \frac 1 {\frac 1 {p_{i | i-1}} + \sum_{k = 1}^K  \frac 1 {\sigma_{W_k}^2}} 
\end{align}
\hspace*{\fill}$\qed$
}

\section{Two equivalent representations of $R_{\mathrm{rm}}(d)$ }
\label{apx:noisy}
In this appendix, we verify that \eqref{eq:noisy} coincides with the lower bound on the causal remote rate-distortion function derived in \cite{kostina2016ratecost}.
Indeed,  \cite[Cor. 1 and Th. 9]{kostina2016ratecost} imply
\begin{align}
 R_{\mathrm{rm}}(d) \geq \frac 1 2 \log \left(a^2 + \frac{\sigma^2_{\sX\| \D \sY^{[K]}} - \sigma^2_{\sX\| \sY^{[K]}} }{d - \sigma_{\sX\| \sY^{[K]}}^2}\right). \label{eq:noisyratecost}
\end{align}
 Here, $\sigma^2_{\sX\| \D \sY^{[K]}} - \sigma^2_{\sX\| \sY^{[K]}}$ is the variance of the innovations of the Gauss-Markov process $\{\bar \sX_i\}$, i.e. 
\begin{align}
 \bar \sX_{i+1} = a \bar \sX_{i} + \bar {\mathsf V}_i, \label{eq:kalmanestimates}
\end{align}
$\bar {\mathsf V}_i \sim \mathcal N(0, \sigma^2_{\sX\| \D \sY^{[K]}} - \sigma^2_{\sX\| \sY^{[K]}})$. The form in \eqref{eq:noisyratecost} leads to that in \eqref{eq:noisy} via \eqref{eq:dprev} and 
\begin{align}
\sigma^2_{\sX\| \D \sY^{[K]}} = a^2  \sigma^2_{\sX\| \sY^{[K]}} + \sigma_{\mathsf V}^2.
\end{align}
\hspace*{\fill}$\qed$

\end{appendices}

\section*{Acknowledgement}
We thank both anonymous reviewers for their insightful and careful reviews, which are reflected in the final version.

\bibliographystyle{IEEEtran}
\bibliography{../../Bibliography/control,../../Bibliography/it,../../Bibliography/vk}

\begin{IEEEbiography}
    [{\includegraphics[width=1in,height=1.25in,clip,keepaspectratio]{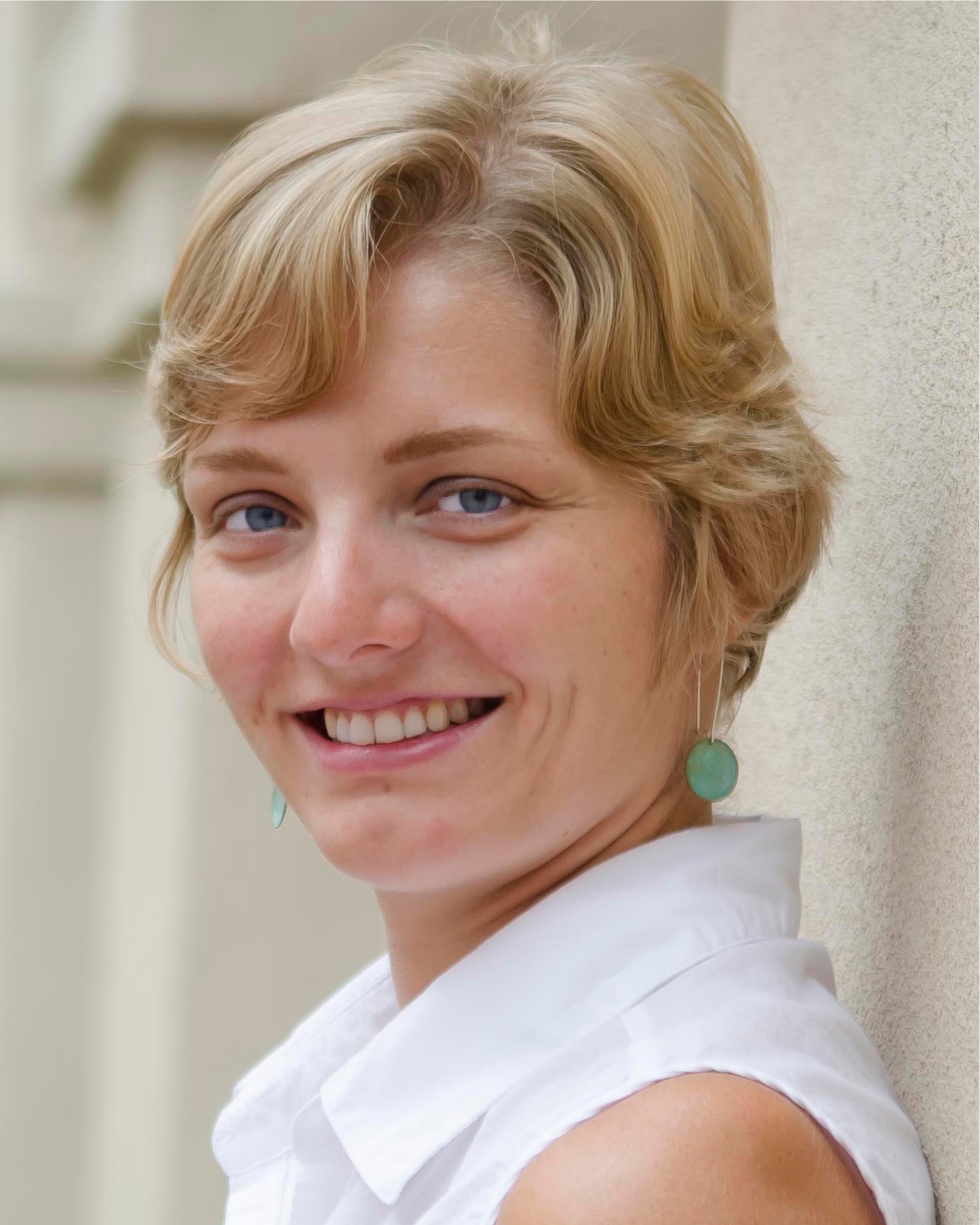}}]{Victoria Kostina}(S'12--M'14)
is a Professor of Electrical Engineering and of Computing and Mathematical Sciences at Caltech. She received a bachelor's degree from Moscow Institute of Physics and Technology (2004), where she was affiliated with the Institute for Information Transmission Problems of the Russian Academy of Sciences, a master's degree from University of Ottawa (2006), and a PhD from Princeton University (2013). She received the Natural Sciences and Engineering Research Council of Canada postgraduate scholarship (2009--2012), the Princeton Electrical Engineering Best Dissertation Award (2013), the Simons-Berkeley research fellowship (2015) and the NSF CAREER award (2017).  Kostina's research spans information theory, coding, control, learning, and communications. 
\end{IEEEbiography}

\begin{IEEEbiography}
    [{\includegraphics[width=1in,height=1.25in,clip,keepaspectratio]{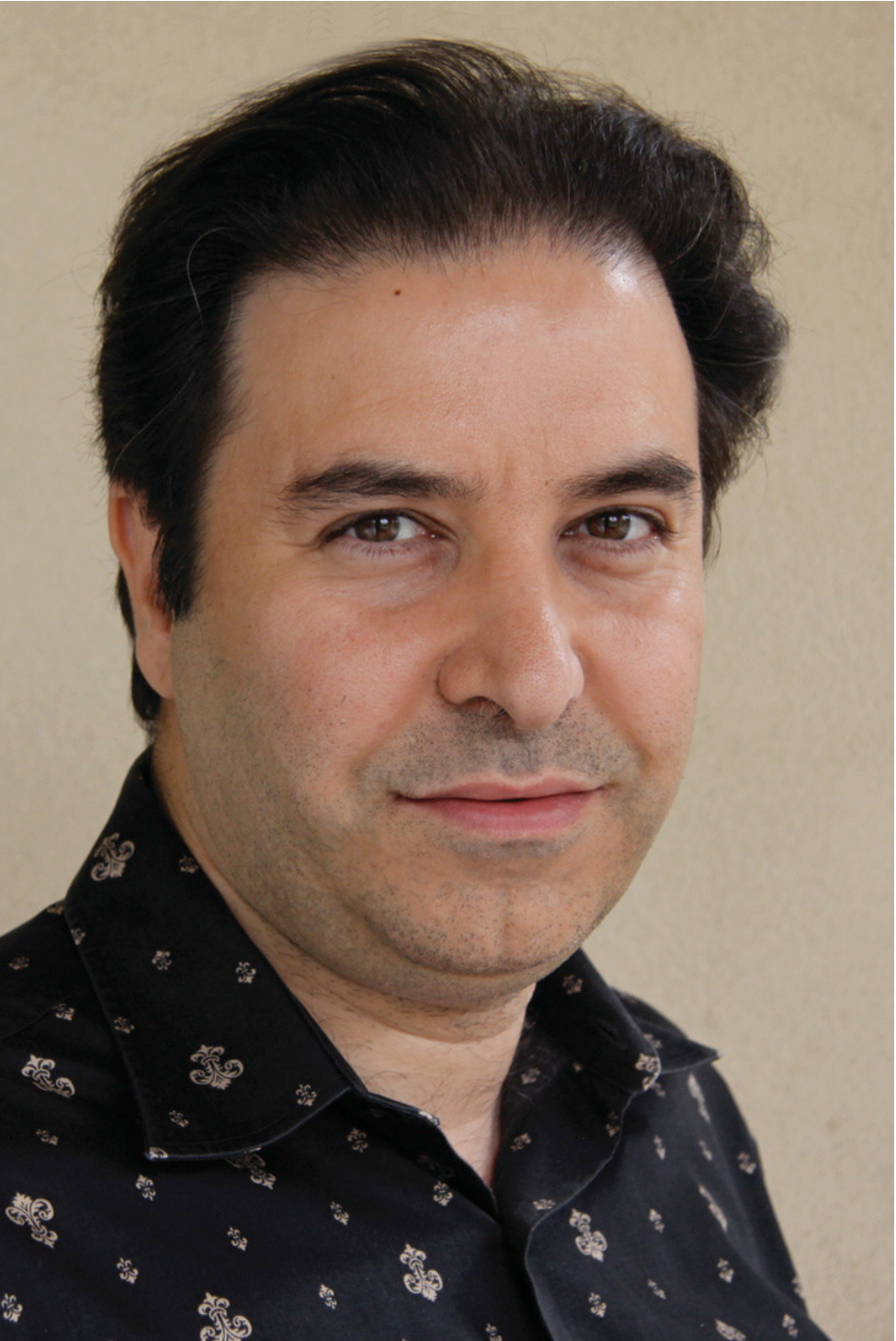}}]{Babak Hassibi}
 was born in Tehran, Iran, in 1967. He received the B.S. degree from the University of Tehran in 1989, and the M.S. and Ph.D. degrees from Stanford University in 1993 and 1996, respectively, all in electrical engineering. 

He has been with the California Institute of Technology since January 2001, where he is currently the Mose and Lilian S. Bohn Professor of Electrical Engineering. From 2013-2016 he was the Gordon M. Binder/Amgen Professor of Electrical Engineering and from 2008-2015 he was Executive Officer of Electrical Engineering, as well as Associate Director of Information Science and Technology. From October 1996 to October 1998 he was a research associate at the Information Systems Laboratory, Stanford University, and from November 1998 to December 2000 he was a Member of the Technical Staff in the Mathematical Sciences Research Center at Bell Laboratories, Murray Hill, NJ. He has also held short-term appointments at Ricoh California Research Center, the Indian Institute of Science, and Linkoping University, Sweden. His research interests include communications and information theory, control and network science, and signal processing and machine learning. He is the coauthor of the books (both with A.H.~Sayed and T.~Kailath) {\em Indefinite Quadratic Estimation and Control: A Unified Approach to H$^2$ and H$^{\infty}$ Theories} (New York: SIAM, 1999) and {\em Linear Estimation} (Englewood Cliffs, NJ: Prentice Hall, 2000). He is a recipient of an Alborz Foundation Fellowship, the 1999 O. Hugo Schuck best paper award of the American Automatic Control Council (with H.~Hindi and S.P.~Boyd), the 2002 National ScienceFoundation Career Award, the 2002 Okawa Foundation Research Grant for Information and Telecommunications, the 2003 David and Lucille Packard Fellowship for Science and Engineering,  the 2003 Presidential Early Career Award for Scientists and Engineers (PECASE), and the 2009 Al-Marai Award for Innovative Research in Communications, and was a participant in the 2004 National Academy of Engineering ``Frontiers in Engineering''program. 

He has been a Guest Editor for the IEEE Transactions on Information Theory special issue on ``space-time transmission, reception, coding and signal processing'' was an Associate Editor for Communications of the IEEE Transactions on Information Theory during 2004-2006, and is currently an Editor for the Journal ``Foundations and Trends in Information and Communication'' and for the IEEE Transactions on Network Science and Engineering. He is an IEEE Information Theory Society Distinguished Lecturer for 2016-2017 and was General Co-Chair if the 2020 IEEE International Symposium on Information Theory (ISIT 2020).
\end{IEEEbiography}

\end{document}